\documentclass[english,11pt,a4paper,dvipsnames]{article}
\ifx\pdfoutput\undefined\else
\pdfoutput=1
\fi 
\usepackage[latin9,utf8]{inputenc}
\usepackage{geometry}
\usepackage{babel}
\usepackage{hyperref}

\usepackage[shortlabels]{enumitem}
\usepackage{float}

\usepackage{amsmath, amssymb, amstext, mathtools}
\usepackage{amsthm}

\usepackage{thm-restate}

\usepackage{algorithm}
\usepackage{algpseudocodex}
\usepackage{tikz}
\usepackage{pgfplots}
\usepackage{textpos}
\usepackage{xcolor}
\usetikzlibrary{fadings,backgrounds,patterns,arrows,decorations.pathreplacing,decorations.pathmorphing,calc}

\usepackage[noabbrev,capitalize]{cleveref}

\newenvironment{claimproof}{\begin{proof}}{\end{proof}}

\newtheorem{theorem}{Theorem}[section]

\newtheorem{lemma}[theorem]{Lemma}
\newtheorem{corollary}[theorem]{Corollary}

\newtheorem{observation}[theorem]{Observation}

\newtheorem{claim}[theorem]{Claim}

\theoremstyle{definition}
\newtheorem{definition}[theorem]{Definition}
\newtheorem{example}[theorem]{Example}
\theoremstyle{remark}
\newtheorem{remark}[theorem]{Remark}

\newcommand{\FA}{{\mathfrak A}}
\newcommand{\FB}{{\mathfrak B}}

\newcommand{\N}{{\mathbb N}}

\newcommand{\R}{{\mathbb R}}
\newcommand{\C}{{\mathbb C}}

\newcommand{\bv}{{\boldsymbol{v}}}
\newcommand{\bw}{{\boldsymbol{w}}}
\newcommand{\bchi}{{\boldsymbol{\chi}}}

\DeclareMathOperator{\ext}{ext}

\newcommand{\trans}{{\top}} 

\newcommand{\zeroMatrix}{\mathbf{0}}

\newcommand{\oneMatrix}[2]{E^{(#1, #2)}}

\newcommand{\cstar}{$\mathrm{C}^\ast$}

\usepackage{microtype}

\newcommand{\block}[3]{\mathbf{#1}[#2, #3]}

\newcommand{\Alg}{\mathbb{A}}
\newcommand{\FullMatrixAlg}{\mathsf{M}}
\newcommand{\ColorMatrix}[1]{M_{#1}}

\newcommand{\Blg}{\mathbb{B}}

\newcommand{\mulSet}[1]{\{\!\{#1\}\!\}}

\newcommand{\kWL}[1]{#1\text{-}\mathrm{WL}}

\newcommand{\cost}{C}

\makeatletter
\def\blfootnote{\gdef\@thefnmark{}\@footnotetext}
\makeatother

\newcommand{\orcid}[1]{\href{https://orcid.org/#1}{\includegraphics[height=1.8ex]{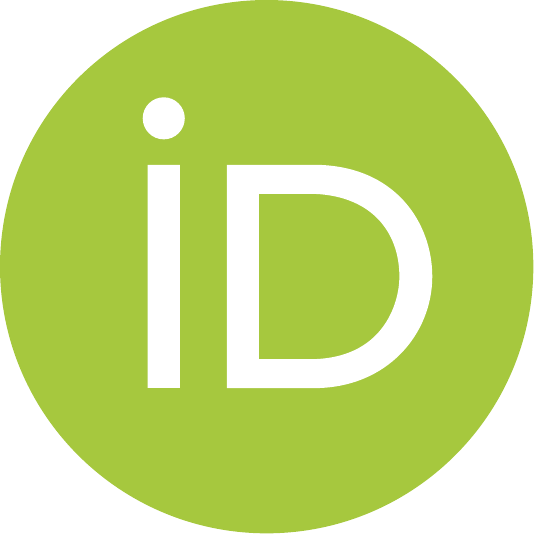}}}

\title{The Classical Weisfeiler-Leman Algorithm\\Stabilizes in $O(n)$ Rounds}

\author{
Simon D\"{o}ring \orcid{0009-0002-6667-5257} \\
Max Planck Institute for Informatics \\ Saarland University (SIC)
\and
Daniel Neuen \orcid{0000-0002-4940-0318}\\
TU Dresden
}
\date{}

\definecolor[named]{urlblue}{cmyk}{1,0.58,0,0.21}
\hypersetup{
    breaklinks=true,
    colorlinks=true,
    citecolor=purple!70!blue!60!black,
    linkcolor=purple!70!blue!90!black,
    urlcolor=urlblue,
    pdflang={en},
    pdftitle={The Classical Weisfeiler-Leman Algorithm Stabilizes in O(n) Rounds},
    pdfauthor={Simon Döring, Daniel Neuen}
}

\begin{document}

\maketitle

\begin{abstract}
 The classical Weisfeiler-Leman algorithm (also known as the $2$-dimensional Weisfeiler-Leman algorithm) is a simple combinatorial algorithm that was originally designed as a heuristic for the graph isomorphism problem.
 However, it has also numerous connections to other areas such as algebraic graph theory, logics, proof complexity, combinatorial optimization and machine learning.

 We prove that the classical Weisfeiler-Leman algorithm terminates after $5(n-1)$ iterations.
 This improves over the previous best upper bound of $O(n \log n)$ by Lichter, Ponomarenko and Schweitzer [LICS 2019], and asymptotically matches the known lower bound of $\Omega(n)$ by Fürer [ICALP 2001].

 Additionally, building on our results for the $2$-dimensional case, we obtain an improved upper bound of $O(n^{k-1}/(k-2)! + n^{k-2})$ on the number of iterations performed by the $k$-dimensional Weisfeiler-Leman algorithm, for every $k \geq 3$.
 Our arguments actually hold for a larger class of sequences of colorings of $k$-tuples; in this larger class our upper bounds are essentially tight for all $k \geq 3$.
\end{abstract}

\section{Introduction}

The Weisfeiler-Leman (WL) algorithm is a simple combinatorial procedure that iteratively computes a coloring of tuples of vertices in a graph, attempting to assign different colors to tuples that are structurally different.
The WL algorithm was originally introduced in its $2$-dimensional version (also known as the \emph{classical Weisfeiler-Leman algorithm}) by Weisfeiler and Leman \cite{WeisfeilerL68}.
For each $k\geq 1$, the $k$-dimensional WL ($k$-WL) algorithm, independently introduced by Immerman and Lander \cite{ImmermanL90} and Babai and Mathon (see \cite{Babai16} for a historical note), colors $k$-tuples of vertices of an input graph $G$.
More precisely, initially the $k$-WL algorithm colors each $k$-tuple $(v_1,\dots,v_k)$ of vertices with the isomorphism type of the underlying (ordered) induced subgraph.
Afterwards, the coloring is iteratively refined by taking colors of $k$-tuples into account that differ in only one coordinate.

The most common application of the WL algorithm is in the context of the graph isomorphism problem.
Indeed, despite the WL algorithm not being a complete isomorphism test for $n$-vertex graphs for dimension $k = o(n)$ \cite{CaiFI92}, it is widely used both in practice \cite{AndersS21,JunttilaK07,JunttilaK11,McKay81,McKayP14} and theory \cite{GroheN21,GroheS20,Kiefer20,Neuen26}.
The most important versions of the WL algorithm are dimensions $1$ and $2$.
The $1$-WL algorithm (also known as \emph{Color Refinement} or \emph{Naive Vertex Refinement}) can be implemented in near-linear time (see, e.g., \cite{BerkholzBG17}) and is already powerful enough to serve as a complete isomorphism test for almost all graphs \cite{BabaiES80}.

Whereas the $1$-dimensional algorithm fails to distinguish between $d$-regular graphs of equal size, the $2$-WL version serves as a complete isomorphism test for almost all regular graphs \cite{Bollobas82}.
Also, the $2$-WL algorithm is closely tied to algebraic graph theory via so-called coherent configurations \cite{Babai95,Higman87}, and is expressive enough to capture, e.g., many spectral invariants \cite{Furer10,RattanS23} as well as other graph properties (see, e.g., \cite{ArvindFKV20,KieferN22,KieferN22b}).
For these reasons, the $2$-WL algorithm is frequently used as a subroutine for isomorphism testing (see, e.g., \cite{ArvindNPZ26,GroheN26,Neuen24a,Neuen24b}) and also forms the basis for extensions such as Deep WL \cite{GroheSW21}.

Higher dimensions of the WL algorithm are often sufficient for a complete isomorphism test on restricted classes of graphs.
For example, it is known that $3$-WL identifies all planar graphs \cite{KieferPS19}, and more generally, for every class excluding some fixed minor , the $k$-WL algorithm is a complete isomorphism test for some fixed $k$ \cite{Grohe17}.
Also, Babai's quasipolynomial-time isomorphism test \cite{Babai16} prominently relies on the $k$-WL algorithm for dimension $k = O(\log n)$.

Beyond that, the WL algorithm has surprising connections to several other areas of (theoretical) computer science.
For example, the expressive power of $k$-WL can be characterized via the $(k+1)$-variable fragment $\mathsf{C}^{k+1}$ of first-order logic with counting quantifiers \cite{ImmermanL90,CaiFI92}, homomorphism counts from graphs of treewidth at most $k$ \cite{Dvorak10,DellGR18}, and suitable linear programming relaxations \cite{AtseriasM13,GroheO15}.
Moreover, the WL algorithm has connections to combinatorial optimization \cite{AtseriasM13,GroheKMS14,Malkin14}, proof complexity \cite{AtseriasF23,BerkholzG15,RezendeFJN025,ToranW24} and machine learning \cite{MorrisLMRKGFB22,MorrisRFHLRG19,ShervashidzeSLMB11,XuHLJ19}.

Besides the dimension of the WL algorithm, the key parameter is the number of iterations it takes until stabilization (see, e.g., \cite{Kiefer20,Kiefer20b}).
The $k$-WL algorithm trivially terminates after $n^{k} - 1$ refinement steps, since the number of different colors of $k$-tuples increases in each iteration.
For the $1$-WL algorithm , this trivial upper bound is in fact tight \cite{KieferM20}, and there are even partial characterizations of those graphs that achieve $n-1$ refinement rounds \cite{KieferM26}.
In contrast, for $k = 2$, the first non-trivial upper bound of $O(n^2/ \log n)$ was given by Kiefer and Schweitzer \cite{KieferS19}.
This bound was improved by Lichter, Ponomarenko and Schweitzer to $O(n \log n)$ \cite{LichterPS19}.
On the other hand, the best known lower bound due to Fürer \cite{Furer01} is $\Omega(n)$.
As the first main result of this work, we close the gap between upper and lower bound for the iteration number of the classical Weisfeiler-Leman algorithm (up to a constant factor).

\begin{restatable}{mtheorem}{mtOne}\label{mtheo:2:dim}
    The 2-WL algorithm stabilizes after at most $5(n-1)$ iterations, where $n$ denotes the number of vertices.
\end{restatable}

We note that our arguments are not limited to the $\kWL{2}$ algorithm.
Indeed, as a conceptual contribution, we introduce the notion of a \emph{multiplicative $2$-refinement sequence} $\bchi = (\chi_0,\dots,\chi_r)$, where $\chi_i \colon V^{2} \to C$ is a coloring of pairs (over some universe $V$ of size $n$) for all $i \in \{0,\dots,r\}$.
As our main technical result, we show that every multiplicative $2$-refinement sequence $\bchi = (\chi_0,\dots,\chi_r)$ has length $r \leq 5(n-1)$.
Since the sequence of colorings produced by the iterative refinement process of $2$-WL is a multiplicative $2$-refinement sequence, the above theorem follows.

We stress that \cref{mtheo:2:dim} is not only relevant to the Weisfeiler-Leman algorithm itself, but has several corollaries via the characterizations mentioned above.
To give just one example, the iteration number of $k$-WL precisely corresponds to the quantifier depth in the corresponding logic $\mathsf{C}^{k+1}$ (see, e.g., \cite{Kiefer20b}).
Exploiting this connection, \cref{mtheo:2:dim} implies that every property expressible in $\mathsf{C}^{3}$ on graphs with $n$ vertices can also be expressed using only quantifier depth $5(n-1)$.

Building on our results for $k = 2$, we also obtain improved bounds for all $k \geq 3$.
For $k \geq 3$, the best known upper bound on the iteration number of $k$-WL has been $O(kn^{k-1} \log n)$ \cite{GroheLN25}.
On the other hand, the currently best known lower bound is $\Omega_k(n^{k/2})$ \cite{GroheLNS25} (which improves over previous lower bounds from \cite{BerkholzN23,GroheLN25}).
As our second main result, we improve the upper bound for all $k \geq 3$.

\begin{restatable}{mtheorem}{mtTwo}\label{mtheo:k:dim}
    For every $k \geq 3$, the $k$-WL algorithm stabilizes after $O\left(\frac{n^{k-1}}{(k-2)!} ~+~ n^{k-2}\right)$ iterations, where $n$ denotes the number of vertices.
\end{restatable}

Note that our upper bounds also significantly improve on the dependence on $k$.
This is particularly relevant when $k$ depends on $n$ (e.g., this is the case in Babai's quasipolynomial-time isomorhism test which uses $k = O(\log n)$).
While the improvement on the dependence on $n$ is minor, let us stress that we view our main contribution to be on the technical side.
Indeed, to obtain the upper bound of $O(kn^{k-1} \log n)$, \cite{GroheLN25} relies on an involved argument via certain matrix tensors obtained from colorings produced by $k$-WL.
In contrast, our upper bound follows from a rather straight-forward application of our results for the case $k = 2$.
Indeed, generalizing the case $k = 2$, we introduce the notion of an \emph{extendable multiplicative $k$-refinement sequence} $\bchi = (\chi_0,\dots,\chi_r)$, where $\chi_i \colon V^{k} \to C$ is a coloring of $k$-tuples (over some universe $V$ of size $n$) for all $i \in \{0,\dots,r\}$.
For such sequences, we can derive an upper bound on $r$ from the corresponding result for $k = 2$ essentially by considering suitable projections of the colorings $\chi_i$.
This argument is significantly simpler and more direct than the arguments used in \cite{GroheLN25}.
Also, it implies that any further improvement for the case $k = 2$ (e.g., on the constant factor) directly translates to an improvement for higher dimensions.

On top of that, for the notion of extendable multiplicative $k$-refinement sequences, it turns out that our upper bounds are essentially optimal.
Indeed, a construction from \cite{GroheLN25} gives extendable multiplicative $k$-refinement sequences\footnote{The notion of an extendable multiplicative $k$-refinement sequence is not present in \cite{GroheLN25}, but it can be easily checked that the constructed sequences satisfy the required properties (in fact, it satisfies even stronger properties).} of length $n^{k-1}/k^{\Omega(k)}$.
Hence, our upper bounds are not only optimal in terms of $n$, but even essentially optimal in terms of $k$.

\paragraph{Related Work.}

There is a significant body of work on the relation between dimension and iteration number of the WL algorithm. 
For example, Grohe and Verbitsky \cite{GroheV06} proved that $O(\log n)$ rounds of $k$-WL can be implemented in $\textsf{TC}^1$.
For many classes of graphs, such as planar graphs \cite{Verbitsky07,GroheK21} or graphs of bounded treewidth \cite{GroheV06}, it turns out that $O(\log n)$ rounds of $k$-WL, for some constant $k$, are sufficient to obtain a complete isomorphism test.
In particular, this implies that isomorphism testing for the corresponding classes of graphs is in $\textsf{TC}^1$.
Additionally, there are connections to descriptive complexity theory \cite{BergeremGKO23}.
Moreover, there has also been a series of results of this flavor in the context of the group isomorphism problem; see, e.g., \cite{GrochowL26}.

\section{Technical Overview}

We give an overview on the proofs of Main Theorems \ref{mtheo:2:dim} and \ref{mtheo:k:dim}.

\subsection{The Classical WL Algorithm}

In the following, we discuss the ideas behind the proof of \cref{mtheo:2:dim}.
Let $G = (V,E)$ be an (undirected) graph on $n$ vertices (our results also apply to more general structures).
The $2$-WL algorithm is defined via a sequence $\bchi = (\chi_0,\dots,\chi_r)$ of colorings $\chi_i\colon V^{2} \to C$ for all $i \in \{0,\dots,r\}$, where $C$ is a suitable set of colors.
The initial coloring $\chi_0$ is defined via $\chi_0(v,v) = 2$ for all $v \in V$, $\chi_0(v, w) = 1$ if $vw \in E$, and $\chi_0(v,w) = 0$ otherwise.

Next, the initial coloring is iteratively refined by setting\footnote{The definition in \cref{sec:WL:algo} is slightly different from this definition. However, it is easy to see that both definitions are equivalent.}
\[\chi_{i+1}(v,w) = \Big\{\!\!\Big\{ \big(\chi(v,x), \chi(x,w)\big) ~\Big|~ x \in V \Big\}\!\!\Big\}\]
for all $i \geq 0$.
This process terminates for the minimal $r \geq 0$ such that $\chi_{r} \equiv \chi_{r+1}$, i.e., the partition of $V^2$ into color classes does not change anymore.
It is not difficult to see that $\chi_0 \succ \chi_1 \succ \dots \succ \chi_r$, i.e., in each step the partition into color classes gets strictly finer.
In particular, we trivially obtain that $r \leq n^2 - 1$ since the number of partition classes increases in each step.

To obtain a linear upper bound on the iteration number $r$, our starting point is similar to \cite{LichterPS19} where each coloring is represented by a set of matrices.
More precisely, for a coloring $\chi\colon V^2 \to C$ and a color $c \in C$, we define the matrix $M_{\chi,c} \in \C^{V \times V}$ via $(M_{\chi,c})_{v,w} = 1$ if $\chi(v,w) = c$, and $(M_{\chi,c})_{v,w} = 0$ otherwise.
Now, let us write $M_i \coloneqq \{M_{\chi_i,c} \mid c \in C\}$ for the set of matrices that occur for the $i$-th coloring.
Also, let us write $\C M_i$ for the $\C$-linear span of $M_i$.
Now, we obtain the following properties:
\begin{enumerate}[label = (\alph*)]
    \item\label{item:matrix-sequence-intro-1} $M_i$ is closed under conjugate transposition,
    \item\label{item:matrix-sequence-intro-2} $\C M_{i-1} \subsetneq \C M_i$, and
    \item\label{item:matrix-sequence-intro-3} $M_{i-1}^{\leq 2} \subseteq \C M_i$
\end{enumerate}
for all $i \in [r]$.
Here, for a set $M \subseteq \FullMatrixAlg_n(\C)$ of $n \times n$ matrices over $\C$, we write
\[M^{\leq k} \coloneqq \{A_1\cdots A_\ell \mid A_i \in M, \ell \leq k\}\]
for all matrix products of length at most $k$ using matrices from $M$.
We call a sequence of such sets $M_0,\dots,M_r \subseteq \FullMatrixAlg_n(\C)$ that satisfy \ref{item:matrix-sequence-intro-1}-\ref{item:matrix-sequence-intro-3}, and such that $I_n \in M_0$, an \emph{$\FullMatrixAlg_n(\C)$-sequence}.
We now consider the more general problem of bounding the length of an $\FullMatrixAlg_n(\C)$-sequence.
Note that, due to \ref{item:matrix-sequence-intro-2}, the dimension of the vector space $\C M_i$ increases in each step, which in particular implies that all $\FullMatrixAlg_n(\C)$-sequences have length at most $n^2-1$.

So let us consider an $\FullMatrixAlg_n(\C)$-sequence $M_0,\dots,M_r \subseteq \FullMatrixAlg_n(\C)$ of maximum length $r$.
Clearly, we have that $\C M_r = \FullMatrixAlg_n(\C)$, since otherwise appending (a generating set for) $\FullMatrixAlg_n(\C)$ results in a longer sequence.

Now, let $\Alg_i \coloneqq \langle M_i \rangle$ denote the closure of $M_i$ under $\C$-scalar multiplication, matrix addition and matrix product.
Note that $\Alg_i$ is also closed under conjugate transposition by \ref{item:matrix-sequence-intro-1}, which implies that $\Alg_i$ is a \cstar-algebra for every $i \in \{0,\dots,r\}$.
Now, let $s \leq r$ be maximum such that $\Alg_s \neq \FullMatrixAlg_n(\C)$.
First, it can be shown that $\C M_s = \Alg_s$ (i.e., $\C M_s$ is closed under matrix multiplication), since otherwise we could construct a longer sequence satisfying \ref{item:matrix-sequence-intro-1}-\ref{item:matrix-sequence-intro-3}. 

At this point, we can rely on the following theorem which shows that $\C M_s = \Alg_s$ has a very specific block form (see \cref{sec:alg:block} for details).
The theorem can be seen as a modification of the Wedderburn–Artin theorem.

\begin{restatable}[{see, e.g., \cite{Davidson96} or \cite[page 218]{HolbrookKLP05}}]{theorem}{theoDirSum}\label{theo:direct:sum}
    Let $\Alg \subseteq \FullMatrixAlg_n(\C)$ be a \cstar-algebra with $I_n \in \Alg$. Then
    there are values $1 \leq \alpha_1, \dots, \alpha_k$ and $n \geq \beta_1 \geq  \dots \geq \beta_k \geq n$ 
    and a unitary matrix $U \in \FullMatrixAlg_n(\C)$ such that
    \[U \Alg U^\ast = \bigoplus_{i = 1}^k I_{\alpha_i} \otimes \FullMatrixAlg_{\beta_i}(\C) \,\text{ and } \, \sum_{i = 1}^k \alpha_i \cdot \beta_i = n.\]
\end{restatable}

By applying the unitary matrix (obtained from the last theorem applied to $\Alg_s$) to all matrices appearing in the sets $M_0,\dots,M_r$, we may simply assume that
\[\C M_s = \Alg_s = \bigoplus_{i = 1}^k I_{\alpha_i} \otimes \FullMatrixAlg_{\beta_i}(\C).\]
Now, for $i = 1,\dots,k$, let $r_i$ denote the maximum length of an $\FullMatrixAlg_{\beta_i}(\C)$ sequence, which is bounded by a linear term in $\beta_i$ by induction.
We show that
\begin{equation}
    \label{eq:overview-1}
    s \leq (k-1) + \sum_{i = 1}^k r_i = (k-1) + \sum_{i=1} O(\beta_i).
\end{equation}

With this, it remains to bound the difference between $r$ and $s$.
Here, we actually prove two bounds.
First, we prove that
\begin{equation}
    \label{eq:overview-2}
    r \leq s + 2 \log(n).
\end{equation}
Indeed, a theorem by Paz (see \cref{theo:dim:algebra}) implies that
\[\FullMatrixAlg_n(\C) = \Alg_{s+1} = \C M_{s+1}^{\leq \lfloor n^{2} / 2 \rfloor}.\]
By repeatedly applying \ref{item:matrix-sequence-intro-3}, we obtain that
\[M_{s+1}^{\leq \lfloor n^{2} / 2 \rfloor} \subseteq M_{s+1 + \log(n^2/2)}\]
which implies that $\C M_{s + 2 \log(n)} = \FullMatrixAlg_n(\C)$.
While this bound is simple and can always be applied, it is itself too weak to obtain the desired linear bound on $r$ and only results in a bound $r = O(n \log n)$ (which would recover the bound from \cite{LichterPS19}).
To address this issue, we show a second, more intricate, bound for the critical case where $\beta_1 > n/2$, which in particular implies that $\alpha_1 = 1$.
In this case, we show that
\begin{equation}
    \label{eq:overview-3}
    r \leq s + 1 + \log(2(n - \beta_1)).
\end{equation}
Now, solving the recursion obtained from Equations \eqref{eq:overview-1}, \eqref{eq:overview-2} and \eqref{eq:overview-3} shows that $r \leq 5(n-1)$.
This also completes the proof of \Cref{mtheo:2:dim}.

Finally, let us stress that the above approach does not only bound the iteration number of the $2$-WL algorithm, but it bounds the length of every sequence $\bchi = (\chi_0,\dots,\chi_r)$, where $\chi_i\colon V^2 \to C$ for every $i \in \{0,\dots,r\}$, such that properties \ref{item:matrix-sequence-intro-1}-\ref{item:matrix-sequence-intro-3} are satisfied for the corresponding sequence of sets of matrices.
We call such coloring sequences multiplicative $2$-refinement sequences over $V$ (we refer to \Cref{sec:col} for details).

\begin{restatable}{theorem}{theoMultSequence}\label{thm:upper-bound-multiplicative}
    Let $\bchi$ be a multiplicative $2$-refinement sequence over $V$, and let $n \coloneqq |V|$.
    Then
    \[|\bchi| \leq 5(n - 1).\]
\end{restatable}

We use \Cref{thm:upper-bound-multiplicative} in our proof of \Cref{mtheo:k:dim}.

\subsection{The $k$-Dimensional Case}

Next, we give an overview on the proof of \cref{mtheo:k:dim}.
Let $G = (V,E)$ be a graph and let $\bchi = (\chi_0,\dots,\chi_r)$ denote the sequence of colorings produced by $k$-WL on input $G$ (see \cref{sec:WL:algo}).
Here, $\chi_i\colon V^{k} \to C$ for all $i \in \{0,\dots,r\}$, i.e., the $k$-WL algorithm colors $k$-tuples of vertices.
Recall that the previous best upper bound obtained in \cite{GroheLN25} is $r = O(kn^{k-1} \log n)$ via a complicated adaption of the arguments from \cite{LichterPS19} that relies on certain matrix tensors.
As our second main contribution, we provide a significantly simpler approach to obtain upper bounds on the iteration number of $k$-WL that, based on our results for $k = 2$, also improves over the previous upper bound for $k \geq 3$.
The key idea is to consider projections of the sequence $\bchi$ to the last two coordinates (we arbitrarily choose to project to the last two coordinates; every choice of two coordinates would work).
More precisely, for a tuple $\bv \in V^{k-2}$, let us define the coloring $\chi_i^{\bv}\colon V^{2} \to C$ via
\[\chi_i^{\bv}(w_1,w_2) = \chi_i(\bv,w_1,w_2)\]
(slightly abusing notation, we write $(\bv,w_1,w_2) \in V^{k}$ for the concatenation of the tuples $\bv$ and $(w_1,w_2)$).
Now, let us consider the sequence
\[\bchi^\bv = (\chi_0^\bv,\dots,\chi_r^\bv)\]
for every $\bv \in V^{k-2}$.
Observe that the sequence $\bchi^\bv$ may contain duplicates, i.e., it may happen that $\chi_i^{\bv} \equiv \chi_{i+1}^{\bv}$ although $\chi_i \not\equiv \chi_{i+1}$.
However, after removing duplicates, one can show that the resulting sequence is a multiplicative $2$-refinement sequence.
Let us write $f(n)$ for the maximal length of a multiplicative $2$-refinement sequence on a universe of size $n$; note that $f(n) = O(n)$ by \cref{thm:upper-bound-multiplicative}.

Now, if we could show that, for each $i \in \{0,\dots,r-1\}$, there is some $\bv \in V^{k-2}$ such that $\chi_i^{\bv} \not\equiv \chi_{i+1}^{\bv}$, then we would be essentially done.
Indeed, in this case, we could associate the $i$-th refinement step with one of the refinement steps in the sequence $\bchi^\bv$ (after removing duplicates), for some $\bv \in V^{k-2}$.
Since there are only $n^{k-2}$ such sequences $\bchi^{\bv}$, and each has length at most $f(n) = O(n)$, this would prove an upper bound of $n^{k-2} \cdot f(n) = O(n^{k-1})$.
In fact, slightly refining this argument, if $\chi_i^{\bv} \not\equiv \chi_{i+1}^{\bv}$, then we also get that $\chi_i^{\bv'} \not\equiv \chi_{i+1}^{\bv'}$ for every $\bv'$ obtained from $\bv$ by permuting coordinates (by exploiting that all colorings are shufflable; see \cref{sec:col} for details).
Taking this fact into account, we would obtain an upper bound of $f(n) \cdot \sum_{j = 1}^{k-2} \binom{n}{j}$.

However, it may happen that there is some $i \in \{0,\dots,r-1\}$ such that $\chi_i^{\bv} \equiv \chi_{i+1}^{\bv}$ for all $\bv \in V^{k-2}$.
To cover this case, we exploit one more property of the sequence of colorings produced by $k$-WL.
We say that $\bchi$ is \emph{extendable} if, for every $i \in \{0,\dots,r-1\}$, every $\bv,\bv' \in V^{k-2}$ and every $w_1,w_2 \in V$ such that
\[\chi_i(\bv,w_1,w_2) \notin \{\chi_{i}(\bv',w_1',w_2') \mid w_1',w_2' \in V\}\]
it holds that
\[\chi_{i+1}(\bv,w_1,w_1) \notin \{\chi_{i+1}(\bv',w_1',w_1') \mid w_1' \in V\}.\]
Being extendable allows us to prove that the situation above can only occur for $O(n^{k-2})$ indices $i \in \{0,\dots,r-1\}$.
This leads to the following theorem, which in particular implies \cref{mtheo:k:dim} by using \cref{thm:upper-bound-multiplicative} to bound the function~$f$.

\begin{restatable}{theorem}{theoHighWL}\label{theo:upper:bound:k}
    Let $k \geq 3$.
    Let $\bchi = (\chi_0,\dots,\chi_r)$ be a multiplicative, extendable $k$-refinement sequence over $V$, and let $n \coloneqq |V|$ with $n > k$.
    Then
    \[|\bchi| \leq 3 \cdot n^{k-2} + f(n) \cdot \sum_{j = 1}^{k-2} \binom{n}{j}.\]
\end{restatable}

\section{Preliminaries}

\subsection{Basics}

We write $\N \coloneqq \{0, 1, 2, \dots\}$ for the set of natural numbers.
For $n \in \N$ we write $[n] \coloneqq \{1, \dots, n\}$.
Also, we always write $\log$ for the logarithm to base 2.

A \emph{(relational) signature} is a set $\sigma = \{R_1,\dots,R_m\}$ where each $R_i$ is a relation symbol with prescribed arity $k_i \geq 1$.
A \emph{$\sigma$-structure} is a tuple $\FA = (V(\FA),R_1^\FA,\dots,R_m^\FA)$ where $V(\FA)$ is a finite universe, and $R_i^\FA \subseteq (V(\FA))^{k_i}$ for every $i \in [m]$.
In the remainder, we usually do not explicitly fix a signature and simply speak of relational structures.
We say a relational structure $\FA = (V(\FA),R_1^\FA,\dots,R_m^\FA)$ has \emph{arity at most $k$} if each relation symbol $R_i$ has arity at most $k$.
For $X \subseteq V(\FA)$, we define the \emph{induced substructure} $\FA[X]$ with vertex set $X$ and relations
\[R_i^{\FA[A]} \coloneqq R_i^{\FA} \cap X^{k_i}\]
where $k_i$ is the arity of $R_i$.
For two relational structures $\FA = (V(\FA),R_1^\FA,\dots,R_m^\FA)$ and $\FB = (V(\FB),R_1^\FB,\dots,R_m^\FB)$, an \emph{isomorhism from $\FA$ to $\FB$} is a bijection $\varphi\colon V(\FA) \to V(\FB)$ such that, for every $i \in [m]$, and every $v_1,\dots,v_{k_i} \in V(\FA)$, it holds that $(v_1,\dots,v_{k_i}) \in R_i^{\FA} \iff (\varphi(v_1),\dots,\varphi(v_{k_i})) \in R_i^{\FB}$.

Observe that every graph $G = (V,E)$ can be viewed as a relational structure using a binary relational symbol $E$.

\subsection{C*-Algebras}\label{sec:c:algebra}

We write $\C$ for the complex numbers.
For a matrix $M \in \C^{n \times m}$, we write $M^\trans$ for the transposition of $M$, and $M^\ast$ for the conjugate transposition of $M$.
We write $\oneMatrix{i}{j}$ for the $n \times m$ matrix whose $(i,j)$-entry is equal to $1$ and all other entries are zero.
These kinds of matrices are commonly known as \emph{matrix units}.
The dimension $n \times m$ of a matrix $\oneMatrix{i}{j}$ is always clear from the context. 

We write $\FullMatrixAlg_n(\C)$ for the algebra of all $n \times n$ matrices over $\C$.
Slightly abusing notation, we also write $\FullMatrixAlg_n(\C)$ for the set of all $n \times n$ matrices over $\C$.
For a set $S \subseteq \FullMatrixAlg_n(\C)$, we write $\C S$ for the $\C$-linear span of $S$, and
\[S^{\leq k} \coloneqq \{M_1\cdot \dots \cdot M_\ell \mid \ell \leq k, M_1,\dots,M_\ell \in S\}\]
for the set of all products of at most $k$ matrices from $S$.
Also, we define $S^\infty \coloneqq \bigcup_{k \geq 1} S^{\leq k}$ for the union over all $S^{\leq k}$.
Moreover, we write
\[\langle S \rangle \coloneqq \C S^\infty\]
for the algebra generated by $S$.

An algebra $\Alg$ over the field $\C$ is called a \cstar-algebra if it comes with an involution $\ast \colon \Alg \to \Alg$ and a norm $\lVert \cdot \rVert$ under which $\Alg$ becomes a Banach algebra.
Additionally, $\lVert \cdot \rVert$ has to be submultiplicative, i.e., $\lVert A B \rVert \leq \lVert A \rVert \cdot \lVert B \rVert$ for all $A,B \in \Alg$. 
In this paper, we only consider matrix-algebras over $\FullMatrixAlg_n(\C)$, and use conjugate transposition as involution~$\ast$ and the standard matrix norm $\lVert \cdot \rVert$ as the norm (i.e., the norm that is induced by the euclidean norm on $\C$).
Note that this norm is submultiplicative.
The following theorem ensures that we do not have to worry about norms and Banach-spaces since all our algebras have finite dimension.

\begin{theorem}\label{theo:subalgebra:cstar}
    Every subalgebra $\Alg$ of $\FullMatrixAlg_n(\C)$ that is closed under taking conjugate transposition is a \cstar-subalgebra. 
\end{theorem}
\begin{proof}
    According to \cite[Corollary 1.4.19]{Megginson98} every finite-dimensional normed space is a Banach space (note that this implicitly assumes that the underlying field is either $\R$ or $\C$). Thus, $\Alg$ is a Banach algebra under $\lVert \cdot \rVert$. Furthermore, since $\Alg$ is closed under taking conjugate transposition, we obtain that $\ast \colon \Alg \to \Alg, X \mapsto X^\ast$ is an involution on $\Alg$. Thus, $\Alg$ is a \cstar-algebra.
\end{proof}

\cref{theo:subalgebra:cstar} ensures that our algebras are always \cstar-algebras as long as we ensure that they are closed under conjugate transposition. This also holds if we use a set of matrices to generate our algebra. For a set $S \subseteq \FullMatrixAlg_n(\C)$, we define $S^\ast \coloneqq \{M^\ast \mid M \in S\}$.

\begin{corollary}\label{theo:semisimple}
    Suppose $S \subseteq \FullMatrixAlg_n(\C)$ is closed under conjugate transposition, i.e. $S^\ast = S$.
    Then $S^{\leq t}$ is also closed under conjugate transposition for all $t \in \N$ and  $\langle S \rangle$ is \cstar-algebra. 
\end{corollary}
\begin{proof}
    Let $S$ be closed under conjugate transposition.
    Then $S^{\leq t}$ and the algebra $\langle S \rangle$ are also closed under conjugate transposition since
    \[\left(\sum_{i = 1}^s \lambda_i \prod_{j = 1}^{t_i} M_{i, j}\right)^{\!\!\ast} = \sum_{i = 1}^s \bar{\lambda}_i \prod_{j = t_i}^1 M_{i, j}^\ast,\]
    where $\prod_{j = k_i}^1  M_{i, j}^\ast$ is the product in reversed order.
    \cref{theo:subalgebra:cstar} implies the second result.
\end{proof}

For two \cstar-algebras $\Alg$ and $\Blg$, an algebra homomorphism $\varphi \colon \Alg \to \Blg$ is a function with $\varphi(A \cdot B) = \varphi(A) \cdot \varphi(B)$ and $\varphi(A + zB) = \varphi(A) + z \varphi(B)$ for all $A, B \in \Alg$ and $z \in \C$.
It is called a $\ast$\nobreakdash-homomorphism if additionally $\varphi(X^\ast) = (\varphi(X))^\ast$ holds.
Note that the image of a $\ast$\nobreakdash-homomorphism is again a \cstar-algebra. 
Moreover, if $\varphi$ is bijective then it is called an $\ast$\nobreakdash-isomorphism.

The following theorem is a weakening of Paz's Theorem \cite{Paz84}.
We use this weakening to simplify some computations later.

\begin{theorem}\label{theo:dim:algebra}
    Let $X \subseteq \FullMatrixAlg_n(\C)$ with $\langle X \rangle = \FullMatrixAlg_n(\C)$ and $n \geq 2$ then $\langle X \rangle = \C X^{\leq \lfloor n^2 / 2 \rfloor }$.
\end{theorem}
\begin{proof}
    By \cite[Theorem 1]{Paz84} we obtain $\langle X \rangle = \C X^{\leq {\lceil(n^2 + 2)/3 \rceil}}$.
    Observe that $\lceil(n^2 + 2)/3 \rceil \leq \lfloor n^2 / 2 \rfloor$ for all natural numbers $n \geq 2$.
\end{proof}

\subsection{Block Representation of Algebras}\label{sec:alg:block}

In this section, we present different ways to represent and combine \cstar-algebras in $\FullMatrixAlg_n(\C)$.

\begin{definition}\label{def:direct:sum}
    For \cstar-algebras $\Alg_1, \dots, \Alg_k$ with $\Alg_i \subseteq \FullMatrixAlg_{\beta_i}(\C)$ and $n \coloneqq \beta_1 + \dots + \beta_k$, we write
    \[\Blg \coloneqq \bigoplus_{i = 1}^k \Alg_i\]
    for the \cstar-algebra of $n \times n$ matrices over $\C$ that consists of \emph{blocks} in $\Alg_i$.
    Formally, we define 
    \[ b_i \coloneqq \sum_{t = 1}^{i-1} \beta_t.\]
    Note that $b_1 = 0$ and $b_{k + 1} = n$.
    Now, $X \in \Blg$ if and only if they are matrices $\mathbf{X}[i] \in \Alg_i$ that satisfy the following property: For all $p, q \in [n]$ let $i, j \in [k]$ be the unique indices with $p \in \left\{b_i + 1, \dots,  b_{i + 1}\right\}$ and  $q \in \left\{b_j + 1, \dots,  b_{j + 1}\right\}$.
    If $i = j$ then $X_{p, q} = \mathbf{X}[i]_{p - b_i, q - b_i}$, otherwise $X_{p, q} = 0$.
    We call $\mathbf{X}[i]$ the $i$th block of $X$.
\end{definition}

\begin{example}
    Consider a matrix $X \in \FullMatrixAlg_2(\C) \oplus \FullMatrixAlg_3(\C) \oplus \FullMatrixAlg_2(\C)$.
    Then $X$ can be written as 
    \[X = 
    \arraycolsep=0.0pt
    \left(\begin{array}{ccc}
    \scalebox{0.8}{\fbox{$\mathbf{X}[1]$}}  \\ 
    & \scalebox{1.0}{$\fbox{$\mathbf{X}[2]$}$} & \\
    & & \scalebox{0.8}{$\fbox{$\mathbf{X}[3]$}$} \\
    \end{array}\right).
    \]
    So $X$ is a block matrix with blocks $\mathbf{X}[1] \in \FullMatrixAlg_2(\C)$,  $\mathbf{X}[2] \in \FullMatrixAlg_3(\C)$, and $\mathbf{X}[3] \in \FullMatrixAlg_2(\C)$ which form the first, second and third block respectively.
    Notice that the blocks are independent of each other and that all entries outside a block are zero.
\end{example}

\begin{definition}\label{def:tensor:product}
    Let $\alpha, \beta$ be two positive integers and $n = \alpha \cdot \beta$.
    We write
    \[\Blg \coloneqq I_\alpha \otimes \FullMatrixAlg_{\beta}(\C)\]
    for the \cstar-algebra of $n \times n$ matrices over $\C$ which consist of $\alpha$ many copies of a \emph{block} in $\FullMatrixAlg_{\beta}(\C)$.
    Formally, we define 
    \[b_i \coloneqq (i-1) \cdot \beta.\]
    Note that $b_1 = 0$ and $b_{k + 1} = n$. 
    Now, $X \in \Blg$ if and only if they are matrices $\block{X}{1}{i} \in \FullMatrixAlg_\beta(\C)$ with $\block{X}{1}{i} = \block{X}{1}{j}$, for all $i, j$, that satisfy the following property:
    For all $p, q \in [n]$ let $i, j \in [k]$ be the unique indices with $p \in \left\{b_i + 1, \dots,  b_{i + 1}\right\}$ and  $q \in \left\{b_j + 1, \dots,  b_{j + 1}\right\}$.
    If $i = j$ then $X_{p, q} = \block{X}{1}{j}_{p - b_j, q - b_j}$, otherwise $X_{p, q} = 0$.
\end{definition}

\begin{example}
    Consider a matrix $X \in I_3 \otimes \FullMatrixAlg_4(\C)$.
    Then $X$ can be written as 
    \[X = 
    \arraycolsep=0.0pt
    \left(\begin{array}{ccc}
    \scalebox{1}{\fbox{$\block{X}{1}{1}$}}  \\ 
    & \scalebox{1}{$\fbox{$\block{X}{1}{2}$}$} & \\
    & & \scalebox{1}{$\fbox{$\block{X}{1}{3}$}$} \\
    \end{array}\right) = \left(\begin{array}{ccc}
    \scalebox{1}{\fbox{$A$}}  \\ 
    & \scalebox{1}{$\fbox{$A$}$} & \\
    & & \scalebox{1}{$\fbox{$A$}$} \\
    \end{array}\right) 
    \]
    So $X$ is a block matrix where a single block $A \in \FullMatrixAlg_4(\C)$ is repeated 3 times.
    This means that the blocks depend on each other with $\block{X}{1}{1} = \block{X}{1}{2} = \block{X}{1}{3}$.
\end{example}

We can also combine both examples as the next example shows.

\begin{example}
    Consider a matrix $X \in (I_2 \otimes \FullMatrixAlg_4(\C)) \oplus (I_3 \otimes \FullMatrixAlg_2(\C))$.
    Then $X$ can be written as 
\[X = 
\arraycolsep=0.0pt
\left(\begin{array}{ccccc}
 \scalebox{1.2}{\fbox{$\block{X}{1}{1}$}}  \\ 
  & \scalebox{1.2}{$\fbox{$\block{X}{1}{2}$}$} & \\
 & & \scalebox{0.9}{$\fbox{$\block{X}{2}{1}$}$} \\
 & & & \scalebox{0.9}{$\fbox{$\block{X}{2}{2}$}$} \\
 & & & & \scalebox{0.9}{$\fbox{$\block{X}{2}{3}$}$}\\
\end{array}\right) 
= \left(\begin{array}{ccccc}
 \scalebox{1.2}{\fbox{$A$}}  \\ 
  & \scalebox{1.2}{$\fbox{$A$}$} & \\
 & & \scalebox{0.9}{$\fbox{$B$}$} \\
 & & & \scalebox{0.9}{$\fbox{$B$}$} \\
 & & & & \scalebox{0.9}{$\fbox{$B$}$}\\
\end{array}\right).
\]
    So $X$ consists of a matrix $A \in \FullMatrixAlg_4(\C)$ that is repeated twice and a matrix $B \in \FullMatrixAlg_2(\C)$ that is repeated three times.
\end{example}

Using the direct sum $\oplus$ and the tensor product $\otimes$, we can formalize the following theorem.

\theoDirSum*

\section{Colorings and Refinements}
\label{sec:col}

In this section, we formally describe the $k$-WL algorithm. After that, we first focus on the case $k = 2$ and associate sets of matrices with the coloring sequence produced by $2$-WL.
Most of the arguments in this section already appeared in previous works \cite{LichterPS19,GroheLN25}, but sometimes formulated in a different way.

\subsection{Basics}

We start by covering basic terminology for colorings.
Fix $k \geq 2$ and let $V$ be a finite set.
In this work, a \emph{($k$-dimensional) coloring on $V$} is a mapping  $\chi\colon V^k \to C$, where $C$ is some non-empty set of colors.
Each coloring $\chi \colon V^k \to C$ induces a partition $\pi(\chi)$ of $V^k$ via its color classes.
For two colorings we write $\chi \succeq \chi'$ if $\pi(\chi')$ is finer than $\pi(\chi)$, i.e., every color class $P' \in \pi(\chi')$ is completely contained in a color class $P \in \pi(\chi$).
Furthermore, we write $\chi \equiv \chi'$ if $\pi(\chi) = \pi(\chi')$.
Note that this is equivalent to $\chi \succeq \chi' \succeq \chi$.
Finally, we write $\chi \succ \chi'$ if $\chi \succeq \chi'$ and $\chi \not\equiv \chi'$.

In this work, we shall only be interested in colorings $\chi\colon V^k \to C$ that are \emph{compatible with equality} and \emph{shufflable} (see also \cite{GroheLN25}).
We say that $\chi$ is \emph{compatible with equality} if, for every $(v_1,\dots,v_k),(w_1,\dots,w_k) \in V^k$ with $\chi(v_1,\dots,v_k) = \chi(w_1,\dots,w_k)$, and every $i,j \in [k]$ it holds that
\[v_i = v_j \iff w_i = w_j.\]
Also, $\chi$ is \emph{shufflable} if for every $(v_1,\dots,v_k), (w_1,\dots,w_k) \in V^k$ and every $\sigma\colon [k] \to [k]$ it holds that
\[\chi(v_1,\dots,v_k) = \chi(w_1,\dots,w_k) \implies \chi(v_{\sigma(1)},\dots,v_{\sigma(k)}) = \chi(w_{\sigma(1)},\dots,w_{\sigma(k)}).\] 
Note that $\sigma$ is not necessarily a permutation.

\begin{definition}
    Let $k \geq 2$ and let $V$ be a finite set.
    A \emph{$k$-refinement sequence over $V$} is a sequence $\bchi = (\chi_0,\dots,\chi_r)$ where $\chi_0,\dots,\chi_r\colon V^k \to C$ are colorings that are shufflable, compatible with equality, and $\chi_0 \succ \chi_1 \succ \dots \succ \chi_r$.
    We define $|\bchi| = r$ as the \emph{length} of the sequence.
\end{definition}

Observe that the length of a $k$-refinement sequence over $V$ is trivially at most $n^k-1$ (where $n \coloneqq |V|$), since the number of partition classes in $\pi(\chi_i)$ strictly increases in each step.

\subsection{The Weisfeiler-Leman Algorithm} \label{sec:WL:algo}

Next, we formally describe the WL algorithm.
Fix some $k \geq 2$ and let $\FA = (V,R_1^{\FA},\dots,R_m^{\FA})$ be a relational structure of arity at most $k$.
On input $\FA$, the $k$-WL produces a sequence of ($k$-dimensional) colorings $\chi_0^\FA,\dots,\chi_r^\FA\colon V^{k} \to C$.

For the initial coloring $\chi_0^\FA$, each $k$-tuple is colored by the isomorphism type of the underlying induced substructure.
More formally, for tuples $(v_1, \dots, v_k),(w_1, \dots, w_k) \in V^{k}$, we have $\chi_0^\FA(v_1, \dots, v_k) = \chi_0^\FA(w_1, \dots, w_k)$ if and only if, for all $i,j \in [k]$, it holds that $v_i = v_j \iff w_i = w_j$ and the mapping defined via $v_i \mapsto w_i$ (which is well-defined by the first property) is an isomorphism from $\FA[v_1,\dots,v_k]$ to $\FA[w_1,\dots,w_k]$.

Now, let $i \geq 0$ and suppose the coloring $\chi_i^\FA\colon V^k \to C$ is already defined.
For $\bv = (v_1,\dots,v_k) \in V^k$ we set
\begin{equation}
    \chi_{i+1}^\FA(\bv) \coloneqq \Big(\chi_i^\FA(\bv), \; \mulSet{(\chi_i^\FA(\bv[w/1]), \dots, \chi_i^\FA(\bv[w/k])) ~|~ w \in V}\Big),
\end{equation}
where $\bv[w/i] \coloneqq (v_1, \dots, v_{i-1}, w, v_{i+1}, \dots, v_k)$, that is the tuple that is obtained by replacing the $i$th entry with $w$.
The notation $\mulSet{\dots}$ denotes a multiset.
Note that each color $\chi_{i+1}^\FA(\bv)$ is a pair of the form $(c, \mulSet{\dots})$ where $c \in C$ and $\mulSet{\dots}$ is a multiset of tuples of length $k$ of colors in $C$.
We call $\mulSet{\dots}$ the \emph{multiset part of $\chi_{i+1}^\FA(\bv)$}.

Observe that $\chi_{i}^\FA \succeq \chi_{i+1}^\FA$ by definition.
Let $r \geq 0$ be minimal such that $\chi_{r}^\FA \succeq \chi_{r+1}^\FA$.
At this point, the refinement procedure terminates and we obtain the sequence $(\chi_0^\FA,\dots,\chi_r^\FA)$.
The following lemma is not difficult to prove.

\begin{lemma}[{\cite[Page 6-9]{GroheLN25}}]\label{lemma:k:WL:refinement}
    Let $k \geq 2$ and $\FA$ a relational structure of arity at most.
    Then the sequence $(\chi_0^\FA,\dots,\chi_r^\FA)$ produced by $k$-WL on input $\FA$ is a $k$-refinement sequence.
\end{lemma}

Let $\bchi = (\chi_0,\dots,\chi_r)$ be a $k$-refinement sequence of a set $V$.
We say that $\bchi$ is a \emph{$k$-WL sequence over $V$} if there is a relational structure $\FA$ with $V(\FA) = V$ such that $\bchi = (\chi_0^\FA,\dots,\chi_r^{\FA})$ is the sequence produced by $k$-WL on input $\FA$.

\subsection{Dimension Two: From Colorings to Matrices}

In the following, we fix $k = 2$ and focus on the proof of \Cref{mtheo:2:dim}.
Our goal is to prove that every $2$-WL sequence $\bchi = (\chi_0,\dots,\chi_r)$ over a set $V$ has length at most $5(n-1)$, where $n \coloneqq |V|$.
Towards this end, we first translate each coloring $\chi_i$ into a set of matrices that span a subspace of $\FullMatrixAlg_n(\C)$.

Given a coloring $\chi \colon V^2 \to C$ and a color $c \in C$, we define the $n \times n$ matrix $\ColorMatrix{\chi, c}$ 
as the matrix with $(\ColorMatrix{\chi, c})_{i, j} = 1$ if $\chi(i, j) = c$ and  $(\ColorMatrix{\chi, c})_{i, j} = 0$ otherwise.
Further, we define $\ColorMatrix{\chi} \coloneqq \{\ColorMatrix{\chi, c} \mid c \in C \}$ and $\langle \chi \rangle = \langle \ColorMatrix{\chi} \rangle$.
The following property relates the colorings in a $2$-WL sequence to multiplications of matrices in $\FullMatrixAlg_n(\C)$.

\begin{definition}\label{def:mul}
    Let $V$ be a finite set of size $n \coloneqq |V|$. We say a $2$-refinement sequence $(\chi_0, \dots, \chi_r)$ over $V$ is \emph{multiplicative} if, 
    for every $i \in \{0, \dots, r-1\}$, we have $\ColorMatrix{\chi_{i}}^{\leq 2} \subseteq  \C \ColorMatrix{\chi_{i+1}}$.
\end{definition}

In order to utilize this property, we have to show that the 2-dimensional Weisfeiler-Leman algorithm produces multiplicative sequences.

\begin{lemma}\label{lem:WL:is:mul}
    Every $2$-WL sequence is multiplicative.
\end{lemma}
\begin{proof}
    The statement is a special case of \cref{lem:k:wl:mul:ref}.
\end{proof}

To use \cref{theo:direct:sum}, we also need that our algebras to contain the identity matrix.
This is ensured by the following lemma.

\begin{lemma}\label{lem:identity:matrix}
    Let $\chi \colon V^2 \to C$ be a coloring that is compatible with equality.
    Then $I_n \in \C M_{\chi}$.
\end{lemma}
\begin{proof}
    Let $a \in V$, $c \coloneqq \chi(a, a)$, and $M \coloneqq M_{\chi, c}$. We show that $M$ is a diagonal 0/1-matrix with $M_{a, a} = 1$. Note that this implies the statement since $I_n$ can then be written as a sum of diagonal matrices $\sum M_{\chi, \chi(a, a)}$.

    To show that $M$ is a diagonal matrix, pick $u,v \in V$ such that $M_{u,v} = 1$.
    Hence, $\chi(u, v) = c = \chi(a, a)$.
    Now, compatibility with equality implies that $u = v$ since $a = a$.
\end{proof}

Also, we need to show we obtain \cstar-algebras which, by Corollary \ref{theo:semisimple}, amounts to showing that the sets $M_\chi$ are closed under conjugate transposition.

\begin{lemma}\label{lem:mul:semisimple}
    Let $\chi \colon V^2 \to C$ be a coloring that is a shufflable.
    Then $\ColorMatrix{\chi}$ is closed under conjugate transposition and $\langle \chi \rangle$ is a \cstar-algebra.
\end{lemma}

\begin{proof}
    We start by showing that when $\chi$ is shufflable
    then $\ColorMatrix{\chi}$ is closed under conjugate transposition.
    For each $M \in \ColorMatrix{\chi}$ there is a color $c$ with $M = \ColorMatrix{\chi, c}$. Thus, $M$ is a $0/1$-matrix. If $M$ is symmetric then $M^\trans = M \in \ColorMatrix{\chi}$. Otherwise, $M$ is not symmetric and we have to show that $M^\trans \in \ColorMatrix{\chi}$. If $M$ is not symmetric, then there are indices $u, v$
    such that $M_{u, v} = 1$ (i.e., $\chi(u, v) = c$) and $M_{v, u} = 0$. Let $c' = \chi(v, u)$ and $M' \coloneqq \ColorMatrix{\chi, c'}$, we show that $M' = M^\trans$.  

    Let $a, b$ be two indices. If $M_{a,  b} = 1$ then $\chi(a, b) = \chi(u, v) = c$.
    Hence, shufflablity yields $\chi(b, a) = \chi(v, u) = c'$ which implies $M'_{b, a} = 1$.
    Next, if $M'_{b, a} = 1$ then  $\chi(b, a) = \chi(v, u) = c'$. Hence, shufflablity yields
    $\chi(a, b) = \chi(u, v) = c$ which implies $M_{a, b} = 1$.
    Thus, $M_{a,b} = 1$ if and only if $M'_{b,a} = 1$ which proves that $M' = M^\trans$ and shows that $\ColorMatrix{\chi}$ is closed under conjugate transposition.
    Now, \cref{theo:semisimple} implies that $\langle \chi \rangle$ is a \cstar-algebra. 
\end{proof}

Each set of matrices $S$ over the index set $V^2$ defines a partition $\pi(S)$ on $V^2$ via the equivalence relation $(u_1, v_1) \sim_S (u_2, v_2)$ if and only if $M_{u_1, v_1} = M_{u_2, v_2}$ for all $M \in S$. Observe that 
\begin{align}\label{eq:pi}
    \pi(\chi) = \pi(M_{\chi}) \quad \text{for all ($2$-dimensional) colorings $\chi$.}
\end{align}
Using this, we obtain our next result, which can be used to show that a sequence is stabilizing by simply comparing the set of matrices $M_{\chi}$ with $M_{\chi'}$. 

\begin{lemma}\label{lem:algera:matrix:sub}
    Let $\chi,\chi' \colon V^2 \to C$ be two colorings over $V$ such that $\chi \succeq \chi'$.
    If $\ColorMatrix{\chi'} \subseteq \C \ColorMatrix{\chi}$, then $\chi \equiv \chi'$.
\end{lemma}
\begin{proof}
    Assume that the statement is false and therefore $\chi \succ \chi'$ (i.e. $\chi'$ is strictly finer than $\chi$). By definition there are
    $u_1, v_1, u_2, v_2$ such that $\chi(u_1, v_1) = \chi (u_2, v_2)$ and $\chi'(u_1, v_1)\neq \chi'(u_2, v_2)$.

    Thus, \cref{eq:pi} implies that there is matrix $A \in \ColorMatrix{\chi'}$
    with $A_{u_1, v_1} \neq A_{u_2, v_2}$. Further, since $\ColorMatrix{\chi'} \subseteq \C \ColorMatrix{\chi}$, we obtain that 
    \[A = \sum_{i = 1}^s \lambda_i B^{(i)}, \text{ where } B^{(i)} \in \ColorMatrix{\chi}.\]
    However, since $A_{u_1, v_1} \neq A_{u_2, v_2}$ there is an $i \in [s]$ 
    with $B^{(i)}_{u_1, v_1} \neq B^{(i)}_{u_2, v_2}$. Hence,  \cref{eq:pi}
    yields  $\chi(u_1, v_1) \neq \chi(u_2, v_2)$ 
    which is a contradiction. Thus, $\chi \equiv \chi'$.
\end{proof}

\section{A Linear Bound on the Iteration Number}

This section contains the core technical contribution of the paper.
We prove \Cref{thm:upper-bound-multiplicative} and deduce \Cref{mtheo:2:dim}.

\subsection{Bounding the Iteration Number via $\Alg$-Sequences}

We start by introducing the notion of an \emph{$\Alg$-sequence} which is used to bound the iteration number.

\begin{definition}
    Let $\Alg \subseteq \FullMatrixAlg_n(\C)$ be a \cstar-algebra with $I_n \in \Alg$.
    We say that $\{I_n\} = M_0 \subseteq \dots \subseteq M_\ell \subseteq \Alg$ is an \emph{$\Alg$-sequence (of length $\ell$)} if 
    \begin{enumerate}[label = (\alph*)]
        \item $M_i$ is closed under conjugate transposition,
        \item $\C M_{i-1} \subsetneq \C M_i$, and
        \item $M_{i-1}^{\leq 2} \subseteq \C M_i$
    \end{enumerate}
    for all $i \in [\ell]$.
    We define $\Lambda(\Alg)$ to be the maximal integer $\ell \geq 0$ such that there is an $\Alg$-sequence $\{I_n\} = M_0 \subseteq \dots \subseteq M_\ell \subseteq \Alg$ of length $\ell$.
\end{definition}

Below, we collect some useful basic properties about $\Alg$-sequences.

\begin{observation}\label{obs:power}
    Let $M_0 \subseteq \dots \subseteq M_\ell$ be an $\Alg$-sequence and $s, r$ integers with $s + r \leq \ell$. Then $\C M_{s}^{\leq 2^r} \subseteq \C M_{s+r}$.
\end{observation}
\begin{proof}
    This is a direct consequence of iterating the rule $M_{i-1}^{\leq 2} \subseteq \C M_i$ and  $(\C X)^{\leq k } = \C (X^{\leq k})$.
\end{proof}

\begin{observation}\label{obs:max:lenght}
    Let $M_0 \subseteq \dots \subseteq M_\ell$ be an $\Alg$-sequence of maximum length.
    Then $\C M_\ell = \Alg$.
\end{observation}
\begin{proof}
    If $\C M_\ell \subsetneq \Alg$, then $M_0 \subseteq \dots \subseteq M_\ell \subseteq \Alg$ is an even longer $\Alg$-sequence. 
\end{proof}

\begin{observation}\label{obs:A:seq:U}
    Let $\Alg \subseteq \FullMatrixAlg_n(\C)$ be a \cstar-algebra with $I_n \in \Alg$, and let $U \in \FullMatrixAlg_n(\C)$ be a unitary matrix.
    If $(M_0,\dots,M_\ell)$ is an $\Alg$-sequence, then $(UM_0 U^\ast,\dots,UM_\ell U^\ast)$ is a $U \Alg U^\ast$-sequence.
    In particular, $\Lambda(\Alg) = \Lambda(U \Alg U^\ast)$.
\end{observation}
\begin{proof}
    First note that $U M_0 U^\ast = \{U I_n U^\ast\} = \{I_n\}$. Next, observe that for every matrix $X$ we obtain $(U X U^\ast)^\ast = U X^\ast U^\ast$. Thus, if $M_i$ is closed under conjugate transposition then $U M_i U^\ast$ is also closed under conjugate transposition. Next, note that $\C M_{i-1} \subsetneq \C M_i$ if and only if $\C (U M_{i-1}U^\ast)\subsetneq \C (U M_i U^\ast)$ since the map $X \mapsto U X U^\ast$ is a $\ast$-isomorphism. Lastly, observe that for all $A, B \in \C M_{i}$ we have $U A U^\ast U B U^\ast = U(A B) U^\ast$ which immediately yields that $M_{i-1}^{\leq 2} \subseteq \C M_i$ implies $(U M_{i-1}U^\ast)^{\leq 2} \subseteq \C (UM_iU^\ast)$. Hence, $(UM_0 U^\ast,\dots,UM_\ell U^\ast)$ is a $U \Alg U^\ast$-sequence. 
\end{proof}

\begin{lemma}\label{lem:max:lenght:algebra}
    Let $M_0 \subseteq \dots \subseteq M_\ell$ be an $\Alg$-sequence of maximum length and $s \in [\ell-1]$ be a position with $\langle M_{s}\rangle \subsetneq \langle M_{s+1} \rangle$ then $\C M_{s} = \langle M_{s}\rangle$. 
\end{lemma}

\begin{proof}
    Assume that $s \in [\ell]$ is a position with $\langle M_{s}\rangle \subsetneq \langle M_{s+1} \rangle$ but $\C M_{s} \subsetneq \langle M_{s}\rangle$. We show this assumption leads to a $\Alg$-sequence of length $\ell +1 $ which is not possible since our sequence is of maximum length.

    For each $j$ we define $V_j = \C M_{s}^{\leq 2^j}$. Let $r$ be the smallest value with $V_r = \langle M_{s}\rangle$. Note that this value is at least $1$ since otherwise $\langle M_s \rangle = \C M_s$. We show that 
    \[M_0, \dots M_{s},  V_1, \dots V_r, \C M_{s+r}, \dots, \C M_\ell \]
    is an $\Alg$-sequence of length $s + r + (\ell - (s+r) + 1) = \ell + 1$. So, by re-index the above sequence we obtain a new sequence $N_0, \dots,  N_{\ell +1}$.
    Next, observe that $\{I_n\} = N_0$, $N_{\ell}= \C M_\ell \subseteq \Alg$ and that each $N_i$ is closed under conjugate transposition. We now consider a case distinction on $i \in [\ell + 1]$ to show $N_{i-1} \subseteq N_i$, $\C N_{i-1} \subsetneq \C N_i$ and $N_{i-1}^{\leq 2} \subseteq \C N_i$.
    \begin{itemize}
        \item If $i \leq s$ then the result directly follows form $M_{i-1} = N_{i-1}$ and $M_{i} = N_i$.
        \item If $i = s+1$ then $\C N_s = \C M_s \subsetneq \C M_s^{\leq 2} = V_1 =  N_{s+1}$ since otherwise $\C M_s = \langle  M_s\rangle$. Further $N_{i-1}^{\leq 2} = M_s^{\leq 2} \subseteq V_1 = \C N_i$ is true by definition. Further, $N_s \subseteq N_{s+1}$ is obvious. 
        \item If $s+2 \leq i \leq s + r $ then $\C N_{i-1} = V_{j-1} \subsetneq V_{j} = \C N_{i}$ for $j = i - s $. Note that we obtain $V_{j-1} \subsetneq V_j$ since otherwise $r = j-1 < r$. Further, $V_{j-1}^{\leq 2} \subseteq \C V_j$ is true by definition.
        \item If $i = s + r + 1$ then $\C N_{i-1} = V_r = \langle M_s \rangle \subsetneq \C M_{s + r} = \C N_i$. To see this, note that $\langle M_s \rangle  = \C M^{\leq 2^r}_s \subseteq \C M_{s + r}$ due to \cref{obs:power}. Further, there is a matrix $X \in \C M_{s+1}$ that is not in $\langle M_s \rangle$ since otherwise $\langle M_s \rangle = \langle M_{s + 1}\rangle$. This yields $\langle M_s \rangle \subsetneq \C M_{s + r}$. 
        Also, $N_{i-1}^{\leq 2} = V_r^{\leq 2} = V_r$, thus $N_{i-1}^{\leq 2} \subseteq \C N_i = N_i$.
        \item If $i \geq s + r +2$ then the result directly follows from $\C M_{i-1} = N_{i-1}$ and $\C M_{i} =  N_i$.
    \end{itemize}
    Hence, $N_0 \subsetneq \dots \subsetneq N_{\ell + 1}$ is a $\Alg$-sequence which is too long. Thus, our assumption that $\C M_{s} \subsetneq \langle M_{s}\rangle$ is wrong.
\end{proof}

The next lemma relates the maximum length of an $\FullMatrixAlg_n(\C)$-sequence to the maximum length of a multiplicative $2$-refinement sequence over an $n$-element universe.

\begin{lemma}\label{lem:lambda:chi}
    Let $\bchi = (\chi_0, \dots, \chi_r)$ be a multiplicative $2$-refinement sequence of colorings over a finite set $V$ of size $n \coloneqq |V|$. Then $r \leq \Lambda(\FullMatrixAlg_n(\C))$.
\end{lemma}
\begin{proof}
    For each refinement $\chi_i$, we obtain a set of matrices $M_i \coloneqq M_{\chi_i}$. We show that $M_i$ is $\FullMatrixAlg_n(\C)$-sequence of length $r$. Without loss of generality, we assume $M_{0} = \{I_n\}$ since otherwise we could extend the sequence (note that $I_n \in \C M_0$ due to \cref{lem:identity:matrix}). Further, since $\bchi$ is a multiplicative $2$-refinement sequence, it follows that $M_{i-1} \subseteq M_i$ and $M_{i-1}^{\leq 2} \subseteq \C M_i$.  Lastly, \cref{lem:mul:semisimple} implies that $M_i$ is closed under conjugate transposition and \cref{lem:algera:matrix:sub} yields $\C M_{i-1} \subsetneq \C M_i$ since otherwise $\chi_{i-1} \equiv \chi_i$.

    Thus, $M_0, \dots, M_{r}$ is an $\FullMatrixAlg_n(\C)$-sequence of length $r$ which yields $r \leq \Lambda(\FullMatrixAlg_n(\C))$.
\end{proof}

Hence, to obtain \Cref{thm:upper-bound-multiplicative}, we now bound the maximum length of an $\FullMatrixAlg_n(\C)$-sequence.

\subsection{The Structure of $\Alg$-Sequences}

To bound the maximum length of an $\Alg$-sequence, we proceed by induction and make use of \Cref{theo:direct:sum}.
The starting point is the following lemma which that allows us to ``decompose'' along direct sums.

\begin{lemma}\label{lemma:lambda:sum}
    Let $\Alg \subseteq M_{n}(\C)$ and $\Blg \subseteq M_{m}(\C)$ be two \cstar-algebras with $I_n \in \Alg$ and $I_m \in \Blg$ then 
    $\Lambda(\Alg \oplus \Blg) \leq \Lambda(\Alg) + \Lambda(\Blg) + 1$.
\end{lemma}
\begin{proof}
    Let $\{I_{n+m}\} = M_0 \subseteq \dots \subseteq M_{\ell} \subseteq A \oplus B$ be a $(\Alg \oplus \Blg)$-sequence of length $\ell$.
    We show that $\ell \leq \Lambda(\Alg)+ \Lambda(\Blg) + 1$.
    Let $\pi^{(\Alg)} \colon \Alg \oplus \Blg \to \Alg$ be the function that takes a matrix in $\Alg \oplus \Blg$ and projects to the $\Alg$ part.
    Note that $\pi^{(\Alg)}$ is a $\ast$-homomorphism.
    Similarly, we define $\pi^{(\Blg)} \colon \Alg \oplus \Blg \to \Blg$ as the projection onto $\Blg$.
    Next, for each $i \in [\ell]$ we define $P_i \coloneq \pi^{(\Alg)}(M_i)$ to be the projection of $P_i$ under $\pi^{(\Alg)}$. Note that $P_0 \subseteq P_1 \subseteq \dots \subseteq P_{\ell}$.

    \begin{claim}\label{claim:PI}
        There are at most $\Lambda(\Alg)$ many positions $i \in [\ell]$ with $\C P_{i-1} \subsetneq  \C P_{i}$.
    \end{claim}
    \begin{claimproof}
        Define $s_0 \coloneqq 0$. Assume that $s_{j-1}$ is defined then we define $s_j$ as the smallest value with $\C P_{s_{j-1}} \subsetneq \C P_{s_{j}}$. This way, we obtain a sequence of values $s_0 \leq \dots \leq s_z \leq \ell$. We show that $P_{s_0} \subseteq \dots \subseteq P_{s_z}$ is a $\Alg$-sequence which immediately implies that $z \leq \Lambda(\Alg)$ and thus there are at most $\Lambda(\Alg)$ many positions with $\C P_{i-1} \subsetneq \C P_{i}$.

        To show that $P_{s_0} \subseteq \dots \subseteq P_{s_z}$ is a $\Alg$-sequence, first note that $P_{s_0} = P_0 = \pi^{(\Alg)}(M_0) = \pi^{(\Alg)}(\{I_{n+m}\}) = \{I_n\}$. Next, every $P_i$ is closed under conjugate transposition since $M_i$ is conjugate transposition and $\pi^{(\Alg)}(A^\ast) = \pi^{(\Alg)}(A)^\ast$. Next, $\C P_{s_{j-1}} \subsetneq \C P_{s_{j}}$ holds by construction. Lastly, for each $i \in [\ell]$ and $A, B \in M_{i-1}$ we have $\pi^{(\Alg)}(A) \cdot \pi^{(\Alg)}(B)  = \pi^{(\Alg)}(AB)$. Now, $AB \in \C M_i$ and the linearity of $\pi^{(\Alg)}$ implies $P^{\leq 2}_{i-1} \subseteq \C P_i$.
    \end{claimproof}

    Next, we define $N_i \coloneqq \{ S \in M_i \mid \pi^{(\Alg)}(S) = \zeroMatrix_n\}$ as the set of matrices of $M_i$ that get mapped to zero. Note that $\C N_i$ is the kernel of $\pi^{(\Alg)}$ when restricted to $\C M_i$.\footnote{We define $\C \emptyset = \{\zeroMatrix\}$, i.e., the vector space that only contains the zero matrix. } Further, we define $T_i = \{I_m\} \cup \pi^{(\Blg)}(N_i)$. Note that $T_0 \subseteq \dots \subseteq T_\ell \subseteq \FullMatrixAlg_{m}(\C)$.

    Each matrix $S_1 \in \Alg \oplus \Blg$ can be interpreted as a pair $(A_1, B_1)$ with $A_1 \in \Alg$ and $B_1 \in \Blg$. For two matrices $S_1, S_2 \in \Alg \oplus \Blg$ we have $S_1 + S_2 = (A_1 + A_2, B_1 + B_2)$, $S_1 \cdot S_2 = (A_1 \cdot A_2, B_1 \cdot B_2)$, and $S_1^\ast = (A_1^\ast, B_1^\ast)$ due to the block structure of matrices in $\Alg \oplus \Blg$. Thus, we can write $N_i = \{(\zeroMatrix_n, S) \in M_i\}$ and $T_i = \{I_m\} \cup \{S \mid (\zeroMatrix_n, S) \in M_i\}$.
    
    \begin{claim}\label{claim:NI}
        There are at most $\Lambda(\Blg) + 1$ many positions $i \in [\ell]$ with $ \C N_{i-1} \subsetneq \C N_{i}$.
    \end{claim}
    \begin{claimproof}
        Define $s_0 \coloneqq 0$. Assume that $s_{j-1}$ is defined then we define $s_j$ as the smallest value with $\C T_{s_{j-1}} \subsetneq \C T_{s_{j}}$. This way we obtain a sequence of values $s_0 \leq \dots \leq s_z \leq \ell$.
        We prove the statement in two steps. 
    
        For the first step, we show that $T_{s_0} \subseteq \dots \subseteq T_{s_z}$ is a $\Blg$-sequence which immediately implies that $z \leq \Lambda(\Blg)$. Observe that $\C T_{i-1} \subsetneq \C T_{i}$ directly implies $\C N_{i-1} \subsetneq \C N_{i}$. 
        We obtain $T_0 = \{I_n\}$ by construction since $N_0 = \emptyset$. Next observer $\pi^{(\Alg)}(A) = \zeroMatrix_n$ if and only if $\pi^{(\Alg)}(A^\ast) = \pi^{(\Alg)}(A)^\ast = \zeroMatrix_n$. Hence, each $N_i$ is closed under conjugate transposition since $M_i$ is closed under conjugate transposition. This implies that $T_i$ is also closed under conjugate transposition. Next, by construction we have $\C T_{s_{j-1}} \subsetneq \C T_{s_{j}}$. Lastly, we show $T_{i-1}^{\leq 2} \subseteq \C T_i$. To this end, it is enough to show that $N_{i-1}^{\leq 2} \subseteq \C N_i$. Let $A, B \in N_{i-1}$ then $A, B \in M_{i-1}$. Hence, $A \cdot B \in \C M_{i}$
        since $M_{i-1}^{\leq 2} \subseteq \C M_i$. Further, $\pi^{(\Alg)}(A \cdot B) = \pi^{(\Alg)}(A) \cdot \pi^{(\Alg)}(B) = \zeroMatrix_n$ which yields $A \cdot B \in \C N_{i}$.
        Hence, $T_{s_0} \subseteq \dots \subseteq T_{s_z}$ is a $\Blg$-sequence.
        
        For the second step, we show that there is at most one position $i \in [\ell]$ with $\C N_{i-1} \subsetneq \C N_{i}$ and $ \C T_{i-1} = \C T_{i}$. Indeed, if this is the case then we have a matrix $(\zeroMatrix_n, S) \notin N_{i-1}$ with $(\zeroMatrix_n, S) \in \C N_{i}$ and $S \in \C T_{i} = \C T_{i-1}$. Thus, $S = I_m$ since otherwise $(\zeroMatrix_n, S)$ would be in $\C N_{i-1}$. Therefore, $(\zeroMatrix_n, I_m) \in \C N_i$. Now, if there is another $j \geq i + 1$ with $\C N_{j-1} \subsetneq \C N_{j}$ and $ \C T_{j-1} = \C T_{j}$ then the same argument yields $(\zeroMatrix_n, I_m) \notin \C N_{j-1}$. However this is not possible since $(\zeroMatrix_n, I_m) \in \C N_i \subseteq \C N_{j-1}$.

        Combining the first step with the second step yields that there are at most $\Lambda(\Blg) +  1$ many positions $i \in [\ell]$ with $\C N_{i-1} \subsetneq \C N_{i}$.
    \end{claimproof}

    \begin{claim}\label{claim:PI:or:NI}
        For every $i \in [\ell]$ we have  $\C P_{i-1}  \neq  \C P_i$ or $\C N_{i-1} \neq \C N_i$.
    \end{claim}
    \begin{claimproof}
        Due to the rank–nullity theorem\footnote{Each linear mapping $\varphi \colon V \to W$ between two finite vector spaces satisfies $\dim(V) = \dim(\varphi(W)) + \dim(\ker \varphi)$ where $\ker \varphi \coloneqq \{x \mid \varphi(x) = 0 \}$ is the kernel of $\varphi$.  }, we obtain
        \[\dim(\C M_i) = \dim(\C P_i) + \dim(\C N_i), \]
        since $\pi^{(A)}$ restricted to $\C M_i$ is a linear map with kernel $\C N_i$ and image $\C P_i$. For $i \in [\ell]$ we have 
        \[ \dim(\C P_{i-1}) + \dim(\C N_{i-1}) = \dim(\C M_{i-1}) < \dim(\C M_{i})  = \dim(\C P_i) + \dim(\C N_i), \]
        thus $\dim(\C P_{i-1}) <  \dim(\C P_i)$ or $\dim(\C N_{i-1}) <  \dim(\C N_i)$. Therefore, we obtain  $\C P_{i-1}  \neq  \C P_i$ or $\C N_{i-1} \neq \C N_i$.
    \end{claimproof}

    Now, let $\{I_{n+m}\} = M_0 \subseteq \dots \subseteq M_{\ell} \subseteq A \oplus B$ be our $(\Alg \oplus \Blg)$-sequence of length $\ell$.
    By \cref{claim:PI:or:NI}, for each $i \in [\ell]$ we have  $\C P_{i-1}  \neq  \C P_i$ or $\C N_{i-1} \neq \C N_i$. \cref{claim:PI} together with \cref{claim:NI} yields that there are at most $\Lambda(\Alg) + \Lambda(\Blg) + 1$ many $i \in [\ell]$ with $\C P_{i-1}  \neq  \C P_i$ or $\C N_{i-1} \neq \C N_i$.
    Thus, $\ell \leq \Lambda(\Alg) + \Lambda(\Blg) + 1$.
\end{proof}

\begin{lemma}
    \label{lem:sequence-direct-sum}
    Let $1 \leq \alpha_1,\dots,\alpha_k \leq n$ and $1 \leq \beta_1,\dots,\beta_k \leq n$ be integers such that $\sum_{i = 1}^k \alpha_i \cdot \beta_i = n$.
    Define
    \[\Alg = \bigoplus_{i = 1}^k I_{\alpha_i} \otimes \FullMatrixAlg_{\beta_i}(\C).\]
    Then
    \[\Lambda(\Alg) \leq (k-1) + \sum_{i = 1}^k \Lambda(\FullMatrixAlg_{\beta_i}(\C)).\]
\end{lemma}
\begin{proof}
    We recursively apply \cref{lemma:lambda:sum} to $\Alg$ which yields 
    \begin{align*}
    \Lambda(\Alg) &= \Lambda\left(\bigoplus_{i = 1}^k I_{\alpha_i} \otimes \FullMatrixAlg_{\beta_i}(\C)\right) \leq  1+ \Lambda\left(\bigoplus_{i = 1}^{k-1} I_{\alpha_i} \otimes \FullMatrixAlg_{\beta_i}(\C)\right)  + \Lambda(I_{\alpha_k} \otimes \FullMatrixAlg_{\beta_k}(\C))  \\
    &\leq \dots \leq  (k-1) +  \sum_{i = 1}^k \Lambda(I_{\alpha_i} \otimes \FullMatrixAlg_{\beta_i}(\C)) = (k-1) +  \sum_{i = 1}^k \Lambda(\FullMatrixAlg_{\beta_i}(\C)).
    \end{align*}
    Here, the last step follows since $ \Lambda(I_{\alpha_i} \otimes \FullMatrixAlg_{\beta_i}(\C)) = \Lambda(\FullMatrixAlg_{\beta_i}(\C))$.
\end{proof}

To be able to bound the length of an $\FullMatrixAlg_n(\C)$-sequence by a linear term in $n$, we require one more auxiliary lemma.

\begin{lemma}\label{lem:mul:bound:beta}
    Let $n \geq 2$, $1 \leq \beta < n$ and $X \subseteq \FullMatrixAlg_n(\C)$ with $I_n\in X$, $X=X^\ast$, $\langle X \rangle = \FullMatrixAlg_n(\C)$,
    \[\mathcal B_\beta \coloneqq \FullMatrixAlg_{\beta}(\C) \oplus \zeroMatrix_{n - \beta} 
    =   
    \left\{
    \begin{pmatrix}
        \;\;C & \zeroMatrix\\
        \;\;\zeroMatrix & \quad \;\;\zeroMatrix_{n - \beta}
    \end{pmatrix}
   \mid C \in \FullMatrixAlg_\beta(\mathbb C)
    \right\} 
   \subseteq X.
    \]
    Then $\C X^{\leq  2(n - \beta)} = \langle X \rangle$.
\end{lemma}
\begin{proof}
    We write $U = \C^\beta \oplus \{0\}^{n- \beta}$ for the vector space of vectors in $\C^n$ whose last $n- \beta$ entries are zero. We define
    \[H_t \coloneqq \C \{A u \mid A \in \C X^{\leq t}, u \in U\}.\]
    Note that $H_0 = U$ since $X^{\leq 0} = \{I_n\}$ and $A H_t \subseteq H_{t+1}$ for all $A \in X$. Next, if $H_t = H_{t + 1}$ then $H_t = H_{s}$ for arbitrarily large $s \geq t$. Hence, $H_t = \langle X \rangle H_t = \FullMatrixAlg_n(\C)H_t$ which immediately yields $H_t = \C^n$ since $U$ contains non-zero vectors. Thus the dimensions of the vector spaces $H_t$ increase strictly until they reach $n$. Since $\dim(H_0) = \beta$, this yields
    \begin{align}\label{eq:H}
        H_{n - \beta} = \C^n \qquad \text{and} \qquad \dim(H_{n - \beta - 1}) \geq n-1
    \end{align}
    Next, we show that
    \begin{align}\label{imp:a:b}
        \text{for all $a, b \in \N$ and $v \in H_a$, $w \in H_b$, we obtain $v w^\ast \in \C X^{\leq a + b + 1}$.}
    \end{align}
    To this end, we write $v = \sum_{i = 1}^c A_i u_i$ and  $w = \sum_{i = 1}^d B_j w_j$ with $A_i \in \C X^{\leq a}$, $B_i \in \C X^{\leq b}$, and $u_i, w_j \in U$. Since all $u_i w_j^\ast$ are in $\mathcal{B}_\beta \subseteq X$ and $\C X^{\leq b}$ is closed under complex conjugation, we obtain
    \[v w^\ast = \sum_{i = 1}^c \sum_{j = 1}^d  A_i (u_i  w_j^\ast) B_j^\ast \in \C X^{ \leq a + b+ 1}, \]
    where we used that $M\cdot N \in X^{\leq x + y}$ for all $M \in X^{\leq x}$ and $N \in X^{\leq y}$. This yields (\ref{imp:a:b}).

    To show $\C X^{\leq 2(n - \beta)} = \FullMatrixAlg_n(\C)$, we prove that there is an orthonormal basis $v_1, \dots, v_n \in \C^n$ such that $v_i v_j^\ast = E^{(i,j)} \in \C X^{\leq 2(n - \beta)}$ for all $i, j \in [n]$. Note that this implies $\C X^{\leq 2(n - \beta)} = \FullMatrixAlg_n(\C)$ since the matrix units clearly generate $\FullMatrixAlg_n(\C)$. We choose, $v_1, \dots, v_n$ in such a way that $v_1, \dots, v_{n-1} \in H_{n - \beta - 1}$ which is possible since $H_{n - \beta - 1}$ has dimension at least $n - 1$ due to (\ref{eq:H}). Due to  (\ref{imp:a:b}) we obtain for all $i \in [n]$ and $j \in [n-1]$ that  $v_i v_j^\ast \in \C X^{\leq 2(n - \beta)}$ and  for all $i \in [n-1]$ and $j \in [n]$ that $v_i v_j^\ast \in \C X^{\leq 2(n - \beta)}$.
    Lastly, $v_n v^\ast_n  = E^{(n,n)} = I_n - \sum_{i =1 }^{n-1} v_i v^\ast_i \in \C X^{\leq 2(n - \beta)}$. Hence, $\C X^{\leq 2(n - \beta)}$ contains all matrix units and is therefore equal to $\FullMatrixAlg_n(\C)$.
\end{proof}

\subsection{Bounding the Length of $\Alg$-Sequences}

To bound the maximum length of an $\FullMatrixAlg_n(\C)$, we use the following recursion with $\cost(1) = 1$ and for $n \geq 2$
\begin{equation}
    \label{eq:def-c}
    \cost(n) \coloneqq \max_{\substack{n > \beta_1 \geq \dots \geq \beta_\ell \geq 1 \\ \beta_1 + \dots + \beta_\ell \leq n \\ \beta_1, \dots, \beta_\ell \in \N}} \; \;   \sum_{i=1}^{\ell} \cost(\beta_i) + \begin{cases}
        2 + \lceil \log(n - \beta_1) \rceil &\text{If $\beta_1 > \frac{1}{2}n$}  \\
        \lceil 2\log(n)\rceil & \text{otherwise}
    \end{cases}.    
\end{equation}

\begin{lemma}\label{lem:c:bound}
    For all $n\geq 1$ we have $\cost(n) \leq 5n - 4$.
\end{lemma}
\begin{proof}
    We proof a slightly stronger statement: $C(n) \leq 5n - q(n)$ for  $q(n) \coloneqq \lceil2 \log(n)\rceil + 4$. To this end, first obverse that 
    \begin{align}\label{eq:q}
        q(a) + q(b) \geq q(a + b) \text{ for all $a, b \geq 1$}.
    \end{align}
    We now proof $C(n) \leq 5n - q(n)$ via strong induction on $\N$. For the base case, observe that $C(1) = 1 = 5\cdot 1 - 0 - 4$. For the induction step, we have to show the statement for $n \geq 2$ while assuming $C(\beta_i) \leq 5\beta_i - q(\beta_i)$ for all $\beta_i < n$. Let $n > \beta_1 \geq \dots \geq \beta_\ell \geq 1$ with $\sum \beta_i \leq n$ and $\beta_i \in \N$ be chosen such that 
    \[C(n) = \sum_{i=1}^{\ell} \cost(\beta_i) + \begin{cases}
    2 + \lceil \log(n - \beta_1) \rceil &\text{If $\beta_1 > \frac{1}{2}n$}  \\
    \lceil 2\log(n)\rceil & \text{otherwise}
\end{cases} .\]
Note that $\sum \beta_i = n$ since otherwise we can add $\beta_{\ell+1} = 1$ to the list to obtain a larger value $C(n)$.  
We now consider a case distinction.

First, assume that $\beta_1 > \frac{1}{2}n$. In this case our induction hypothesis yields
\begin{align*}
    C(n) \leq 5n - \sum_{i = 1}^\ell q(\beta_i) +  2 + \lceil \log(n - \beta_1) \rceil.
\end{align*}
By applying \eqref{eq:q} multiple times, we obtain $\sum_{i = 2}^r q(\beta_i) \geq q(n - \beta_1)$. Using this, we obtain
\begin{align*}
C(n) &\leq 5n - q(n- \beta_1) - q(\beta_1) +  2 + \lceil \log(n - \beta_1) \rceil \\ &= 5n + \lceil \log(n - \beta_1) \rceil - \lceil 2\log(n - \beta_1) \rceil - \lceil 2\log(\beta_1) \rceil   - 6 \\
&\leq 5n - \lceil 2\log(\beta_1) \rceil   - 6,
\end{align*}
where we used $\lceil \log(n - \beta_1) \rceil \leq \lceil 2\log(n - \beta_1) \rceil$ for the last step. Further, since $\beta_1 > \frac{1}{2}n$, we obtain that $\lceil2 \log(\beta_1)\rceil     \geq \lceil 2(\log(n/2) )  \rceil   = \lceil 2(\log(n) -1)  \rceil = \lceil 2 \log(n)  \rceil - 2$. Using this yields
\[C(n) \leq 5n - (\lceil 2\log(n) \rceil  - 2) - 6 \leq 5n -  \lceil2 \log(n)\rceil - 4 = 5n - q(n).\]

Second, we assume $\beta_i \leq \frac{1}{2}n$ for all $i \in [\ell]$. In this case our induction hypothesis yields
\begin{align*}
    C(n) \leq 5n - \sum_{i = 1}^\ell q(\beta_i) +   \lceil 2\log(n) \rceil
\end{align*}
We consider a case distinction. First assume that $\ell = 2$. Now, $\beta_1 = \beta_2 = n/2$. We use that $\lceil 2 \log(n/2) \rceil  + 2 = \lceil 2 (\log(n/2) + 1) \rceil =\lceil 2 \log(n) \rceil$ to obtain
\[C(n) \leq 5n - 2(\lceil 2 \log(n/2) \rceil + 2 + 2) +   \lceil 2\log(n) \rceil = 5n - 2 \lceil 2 \log(n) \rceil  - 4   +   \lceil 2\log(n) \rceil,  \]
which proves the induction step.
Next, assume $\ell \geq 3$. We start by showing the following claim.
\begin{claim}
    We can partition $\frac{1}{2}n \geq \beta_1 \geq \dots \geq \beta_\ell \geq 1$ into three sets $A_1, A_2, A_3$ with $\gamma_i \coloneqq \sum_{\beta \in A_i} \beta$ such that $\gamma_1, \gamma_2 \geq n/4$ and $\gamma_3 > 0$.
\end{claim}
\begin{claimproof}
First, if $\beta_2 \geq \frac{1}{4}n$ then we are immediately done since we can choose $A_1 = \{1\}, A_2 = \{2\}, A_3 = \{3, \dots ,\ell\}$. Now, assume $\beta_i < \frac{1}{4}n$ for all $i \geq 2$. Let $t$ be the smallest value with $\gamma_1 \coloneqq \sum_{i = 1}^t \beta_i \geq \frac{1}{4}n$ and let $s$ be the smallest value with $\gamma_2 \coloneqq \sum_{i = t+1}^s \beta_i \geq \frac{1}{4}n$. It is sufficient to show that these values exist and $s < \ell$ since this yields $A_1 = \{1, \dots, t\}, A_2 = \{t+1, \dots, s\}, A_3 = \{s+1, \dots ,\ell\}$. We consider a case distinction. If $\beta_1 \geq \frac{1}{4}n$ then $t = 1$. Further, $s < \ell$ since $n - \beta_1 \geq \frac{1}{2}n$ and $\beta_i < \frac{1}{4}n$ for all $i \geq 2$. Otherwise, all $\beta_i < \frac{1}{4}n$ for all $i \in [\ell]$. In this case, $n - \gamma_1 > \frac{1}{2}n$ and thus $n - \gamma_1 - \gamma_2 > 0$.
\end{claimproof}
    
Using this claim we obtain $A_1, A_2, A_3$ and $\gamma_1, \gamma_2, \gamma_3$. Note that for $\gamma \geq n/4$ we obtain
\begin{align}\label{eq:gamma}
    q(\gamma ) = \lceil 2 \log(\gamma) \rceil  + 4 = \lceil 2 (\log(\gamma) + 2) \rceil \stackrel{(4\gamma \geq n)}{\geq } \lceil 2 \log(n) \rceil
\end{align}
Using this we obtain 
\begin{align*}
    C(n) &\leq 5n - \sum_{i = 1}^\ell q(\beta_i) +   \lceil 2\log(n) \rceil  \\
    &\stackrel{(\ref{eq:q})}{\leq} 5n - q(\gamma_1) - q(\gamma_2) - q(\gamma_3)  +  \lceil 2\log(n) \rceil \\
    &\stackrel{(\ref{eq:gamma})}{\leq} 5n - \lceil 2\log(n) \rceil  - q(\gamma_3) \leq 5n-\lceil 2\log(n) \rceil -4. \qedhere
\end{align*}
\end{proof}

Now, we can bound the maximum length of an $\FullMatrixAlg_n(\C)$-sequence by $C(n) - 1$, which is linear in $n$ by the previous lemma.

\begin{theorem}\label{theo:lambda:C}
    Let $n \geq 1$.
    Then $\Lambda(\FullMatrixAlg_n(\C)) \leq  \cost(n) - 1$.
\end{theorem}

\begin{proof}
    We prove the statement by induction on $n \geq 1$.
    For $n = 1$ we obtain that $\Lambda(\FullMatrixAlg_n(\C)) = 0 \leq  \cost(1) - 1$ as desired.
    So, assume that $n \geq 2$.
    
    Let $\{I_n\} = M_0 \subseteq \dots \subseteq M_\ell \subseteq \FullMatrixAlg_n(\C)$ be an  $\FullMatrixAlg_n(\C)$-sequence of length $\ell = \Lambda(\FullMatrixAlg_n(\C))$, i.e., for all $i \in [\ell]$ we have
    \begin{itemize}
        \item $M_i$ is closed under conjugate transposition,
        \item $\C M_{i-1} \subsetneq \C M_i$, and
        \item $M_{i-1}^{\leq 2} \subseteq \C M_i$.
    \end{itemize}

    For $i \in [\ell]$, let $\Alg_i = \langle M_i \rangle$ be the \cstar-algebra generated by $M_i$.
    Observe that $\Alg_0 \subseteq \dots \subseteq \Alg_\ell$, $I_n \in \Alg_i$, and $\C M_i \subseteq \Alg_\ell$ for all $i \in [\ell]$. Further, \cref{obs:max:lenght} implies $\Alg_\ell = M_n(\C)$.
    Let $s \in \{0,\dots,\ell\}$ denote the maximal index such that $\Alg_s \neq \FullMatrixAlg_n(\C)$ (observe that such an index exists since $M_0 = \{I_n\}$ and $n \geq 2$).
    We define $\Alg \coloneqq \Alg_s$.
    Observe that $s \leq \Lambda(\Alg)$ since $M_0,\dots,M_s$ is an $\Alg$-sequence.

    By Theorem \ref{theo:direct:sum}, there are integers $1 \leq \alpha_1,\dots,\alpha_k \leq n$ and $n \geq \beta_1 \geq \dots \geq \beta_k \geq 1$, and a unitary matrix $U \in \FullMatrixAlg_n(\C)$ such that
    \[U \Alg U^\ast = \bigoplus_{i = 1}^k I_{\alpha_i} \otimes \FullMatrixAlg_{\beta_i}(\C) \quad\text{and}\quad \sum_{i = 1}^k \alpha_i \cdot \beta_i = n.\]
    Since by \cref{obs:A:seq:U} the sequence $(UM_0 U^\ast,\dots,UM_\ell U^\ast)$ is still a $\FullMatrixAlg_n(\C)$-sequence of length $\ell = \Lambda(\FullMatrixAlg_n(\C))$, we may assume without loss of generality that $U = I_n$ and hence $\Alg$ is already in the form described above.
    Also note that $\beta_1,\dots,\beta_k < n$, since $\Alg \neq \FullMatrixAlg_n(\C)$.
    
    Now, \cref{lem:sequence-direct-sum} together with the induction hypothesis yields
    \begin{align}\label{eq:s}
        s \leq \Lambda(\Alg) \stackrel{(\ref{lem:sequence-direct-sum})}{\leq} (k-1) + \sum_{i = 1}^k \Lambda(\FullMatrixAlg_{\beta_i}(\C)) \leq \left(\sum_{i = 1}^k \cost(\beta_i)\right) - 1.
    \end{align}

    It remains to bound the difference between $\ell$ and $s$. We first prove the following basic bound.
    \begin{claim}\label{claim:bound:big}
        $\ell \leq s + \lceil2 \log(n)\rceil$.
    \end{claim}
    \begin{claimproof}
        Let $t = s + \lceil2 \log(n)\rceil$.
        We assume that $t < \ell$ and derive a contradiction. By assumption we have $\langle M_{s+1} \rangle = \Alg_{s+1} = \Alg_\ell$. We now show that  
        \begin{align}\label{eq:dim:case}
            \Alg_\ell = \langle M_{s+1} \rangle = \C M_{s+  \lceil2 \log(n)\rceil} =  \C M_{t}
        \end{align}
        implies the result. Assuming \eqref{eq:dim:case}, we obtain $\C M_t = \C M_{t+1}$ since otherwise $\Alg_\ell \subseteq \C M_{t} \subsetneq \C M_{t+1} \subseteq \Alg_\ell$ which is a contradiction. However, $\C M_t = \C M_{t+1}$ is also not possible since the sequence is a $\FullMatrixAlg_n(\C)$-sequence. Thus, our assumption that $t < \ell$ is wrong proving the claim.

        It remains to show \eqref{eq:dim:case}. To this end, we use $
        \C M_{s+1}^{\leq \lfloor n^2 / 2 \rfloor} = \langle M_{s+1}\rangle = \FullMatrixAlg_n(\C)$
        due to \cref{theo:dim:algebra}. This yields
        \[\Alg_\ell \subseteq \C M_{s+1}^{\leq \lfloor n^2 / 2 \rfloor} \stackrel{(\ref{obs:power})}{\subseteq} \C M_{s+1 + \lceil \log(n^2 / 2) \rceil} = \C M_{t} \subseteq \Alg_\ell.\qedhere\]
    \end{claimproof}

    Next, we prove the following more intricate bound which applies if $\beta_1$ is sufficiently large.

    \begin{claim}\label{claim:bound:small}
        Suppose $\beta_1 > \frac{1}{2}n$.
        Then $\ell \leq s + 2 + \lceil\log(n - \beta_1)\rceil$.
    \end{claim}
    \begin{claimproof}
         First observe that $\alpha_1 = 1$, since $\sum_{i = 1}^k \alpha_i \cdot \beta_i = n$. Define $\mathcal{B}_{\beta_1} \coloneqq \FullMatrixAlg_{\beta_1}(\C) \oplus \zeroMatrix_{n - \beta_1}$. By \cref{lem:max:lenght:algebra}, we obtain $\C M_s = \Alg_s$ and thus $\mathcal{B}_{\beta_1} \subseteq \C M_s$.

        Let $t = s + 2 + \lceil\log(n - \beta_1)\rceil$.
        We assume that $t < \ell$ and derive a contradiction. By assumption we have $\langle M_{s+1} \rangle = \Alg_{s+1} = \Alg_\ell$. We now show that  
        \begin{align}\label{eq:dim:case:2}
            \Alg_\ell = \langle M_{s+1} \rangle = \C M_{s+2 + \lceil\log(n - \beta_1)\rceil} =  \C M_{t}
        \end{align}
        implies the result. Assuming (\ref{eq:dim:case:2}), we obtain $\C M_t = \C M_{t+1}$ since otherwise $\Alg_\ell \subseteq \C M_{t} \subsetneq \C M_{t+1} \subseteq \Alg_\ell$ which is a contradiction. However, $\C M_t = \C M_{t+1}$ is also not possible since the sequence is a $\FullMatrixAlg_n(\C)$-sequence. Thus, our assumption that $t < \ell$ is wrong proving the claim.

        It remains to show (\ref{eq:dim:case:2}).
        To this end, we use $\C M_{s+1}^{\leq 2(n - \beta_1)} = \langle M_{s+1}\rangle = \Alg_{\ell}$ 
        due to \cref{lem:mul:bound:beta} since $\Alg_\ell = \FullMatrixAlg_n(\C)$, $I_n \in \C M_{s+1}$,  $\C M_{s+1}$ is closed under conjugate transposition, and $\mathcal{B}_{\beta_1} \subseteq \C M_{s}$. This yields
        \[\Alg_\ell \subseteq  \C M_{s+1}^{\leq  2(n - \beta_1)} \stackrel{(\ref{obs:power})}{\subseteq} \C M_{s+2 + \lceil \log(n - \beta_1)\rceil} = \C M_{t} \subseteq \Alg_\ell.\qedhere\]
    \end{claimproof}

    Lastly, we consider two cases.
    First assume that  $\beta_1 \leq \frac{1}{2}n$.
    Then \cref{claim:bound:big} implies 
    \[\ell \leq s +  \lceil 2 \log(n) \rceil \stackrel{\eqref{eq:s}}{ \leq }  \left(\sum_{i = 1}^k \cost(\beta_i)\right) - 1 + \lceil2  \log(n) \rceil \leq C(n) - 1. \]
    Otherwise, $\beta_1 > \frac{1}{2}n$. Thus, \cref{claim:bound:small} yields
    \[\ell \leq s + 2 + \lceil\log(n - \beta_1)\rceil \stackrel{\eqref{eq:s}}{ \leq }  \left(\sum_{i = 1}^k \cost(\beta_i)\right) - 1 +  2 + \lceil\log(n - \beta_1)\rceil \leq C(n) - 1 .\]
    Hence, $\ell \leq  C(n) - 1$ which proves the induction step and thus $\Lambda(\FullMatrixAlg_n(\C)) \leq  C(n) - 1$.
\end{proof}

\theoMultSequence*
\begin{proof}
    By combing \cref{lem:lambda:chi} with \cref{lem:c:bound} and \cref{theo:lambda:C}, we obtain \[|\bchi| \leq \Lambda(\FullMatrixAlg_n(\C)) \leq C(n) - 1 \leq 5(n-1). \qedhere\]
\end{proof}

\mtOne*
\begin{proof}
    Each run of the $2$-WL algorithm produces a sequence of refinements $\bchi$. By \cref{lemma:k:WL:refinement,lem:WL:is:mul}, this sequence is a multiplicative 2-refinement sequence over $V$. Hence, \cref{thm:upper-bound-multiplicative} yields the result.
\end{proof}

\begin{remark}
    Assuming the Paz's conjecture, we obtain that $\langle X \rangle = \C X^{\leq 2n - 2}$ for all $X \subseteq \FullMatrixAlg_n(\C)$ with $\langle X \rangle = \FullMatrixAlg_n(\C)$. Using this, we can modify our cost function $\cost(n)$ to obtain a new function:
    \[D(n) \coloneqq \max_{\substack{n > \beta_1 \geq \dots \geq \beta_\ell \geq 1 \\ \beta_1 + \dots + \beta_\ell \leq n \\ \beta_1, \dots, \beta_\ell \in \N}} \; \;   \sum_{i=1}^{\ell} D(\beta_i) + \begin{cases}
        2 + \lceil \log(n - \beta_1) \rceil &\text{If $\beta_1 > \frac{1}{2}n$}  \\
        1 + \lceil\log(2n - 2)\rceil & \text{otherwise}
    \end{cases}.
    \]
    Again, we can show $\Lambda(\FullMatrixAlg_n(\C)) \leq D(n) - 1$ by simulating the proof of \cref{theo:lambda:C}.
    Furthermore, one can show that $D(n) \leq 4.5 n - 4$ for $n\geq 2$.
    Thus, assuming  Paz's conjecture, we obtain that the classical Weisfeiler-Leman algorithm stabilizes after at most $4.5 n - 5$ many rounds. 
\end{remark}

\section{Higher Dimensions}

In this section, we prove \Cref{mtheo:k:dim}. As in the previous section, we actually prove a more general result on the length on certain $k$-refinement sequences ($k \geq 3$) which, in particular, include all $k$-WL sequences.

\subsection{Definitions and Basic Properties}

For the remainder of this section, let us fix some $k \geq 3$.
The proof of \Cref{mtheo:k:dim} is based on \Cref{thm:upper-bound-multiplicative}.
In order to turn a $k$-refinement sequence into suitable $2$-refinement sequences, we use projections.
Let $V$ be a finite set and let $\chi\colon V^k \to C$ be a coloring of $k$-tuples over $V$.
For $\bv = (v_1,\dots,v_{k-2}) \in V^{k-2}$ we define $\chi^{\bv}\colon V^2 \to C$ via
\[\chi^{\bv}(w_1,w_2) \coloneqq \chi(\bv,w_1,w_2) =  \chi(v_1, \dots, v_{k-2},w_1,w_2)\]
as the \emph{$\bv$-restriction of $\chi$.}

\begin{definition}
    Let $\bchi = (\chi_0,\dots,\chi_r)$ be a $k$-refinement sequence over $V$.
    We say that $\bchi$ is \emph{multiplicative} if $\bchi^\bv = (\chi_0^\bv,\dots,\chi_r^\bv)$ is multiplicative for every $\bv \in V^{k-2}$.
\end{definition}

We start by showing that every $k$-WL sequence is multiplicative.
This also includes a proof for \Cref{lem:WL:is:mul}. Hence, we allow $k = 2$. Note that \cref{lemma:k:WL:refinement} already yields that each  $k$-WL sequence over $V$ is a $k$-refinement sequence over $V$.

\begin{lemma}\label{lem:k:wl:mul:ref}
    Let $k \geq 2$,  $V$ be a finite set and let $\bchi = (\chi_0,\dots,\chi_r)$ be a $k$-WL sequence over $V$.
    Then $\bchi$ is multiplicative.
\end{lemma}

\begin{proof}
    Let $\bv \in V^{k-2}$ and $i \in \{0, \dots, r-1\}$.
    Note that for $k = 2$, we have that $\bv \in V^0$ is the \emph{empty vector}, meaning that $\chi_i^\bv = \chi_i$.
    We need to check that
    \[M_{\chi_i^\bv}^{\leq 2} \subseteq \C M_{\chi_{i+1}^\bv}.\] 
    For simplicity, let us write $\chi \coloneqq \chi_i^\bv$ and $\chi' \coloneqq \chi_{i+1}^\bv$.
    Now, for each $X  \in M_{\chi}^{\leq 2} $ there are two colors $a, b$ such that $A \coloneqq M_{\chi, a}$, $B \coloneqq M_{\chi, b}$, and $X = A \cdot B$.
    Let $C$ and $C'$ denote the set of colors in the image of $\chi$ and $\chi'$, respectively.
    By construction, the multiset part of each color in $C'$ is a multiset of tuples of colors in $C$.
    For every $q \geq 0$, we define $D'_q \subseteq C'$ to be the set of all colors in $C'$ that contain a color from $C_{b,a} \coloneqq \{(c_1, \dots, c_{k-2}, b, a) \mid c_1, \dots, c_{k-2} \in C\}$ inside their multiset part exactly $q$ times (if $k  = 2$ then $C_{b, a} \coloneqq \{(a, b)\}$). 
    We show that 
    \[X = \sum_{q = 0}^\infty \sum_{d \in D'_q} q \cdot  M_{\chi', d},\]
    which yields that $X \in \C M_{\chi'}$. For arbitrary vertices $u, u'$, we obtain that there is exactly one color $d \coloneqq \chi'(u, u') \in C'$ such that $(M_{\chi', d})_{u, u'} = 1$, and $(M_{\chi', d'})_{u, u'} = 0$ for all other colors $d' \neq d$.
    Thus, it is enough to show that for $d \in D'_q$ we have $q = X_{u, u'}$.
    Observe that
    \begin{align*}
    X_{u, u'} = \sum_{w \in V} A_{u, w} \cdot B_{w, u'} &=  |\{w \in V \mid \chi(w, u') = b \text{ and }  \chi(u, w) = a\}| \\
    &= |\{w \in V \mid (\chi_i((\bv, u, u')[w/1]), \dots,  \chi_i((\bv, u, u')[w/k]) \in C_{b,a} \}|.
    \end{align*} 
    By definition the multiset part of $\chi'(u, u')$ is equal to 
    \[\mulSet{ (\chi((\bv, u, u')[w/1]), \dots,  \chi((\bv, u, u')[w/k]) ~|~ w \in V }.\]
    Thus, the multiset part of $d = \chi'(u, u')$ contains the colors of $C_{b, a}$ exactly $X_{u, u'}$ times which implies that $d \in \mathcal{D}'_q$ for $q = X_{u, u'}$.
\end{proof}

To obtain the desired bound on the length of a $k$-refinement sequence, we need a second property.

\begin{definition}
    Let $\bchi = (\chi_0,\dots,\chi_r)$ be a $k$-refinement sequence over $V$.
    We say that $\bchi$ is \emph{extendable} if, for every $i \in \{0,\dots,r-1\}$, every $\bv,\bv' \in V^{k-2}$ and every $w_1,w_2 \in V$ such that
    \[\chi_i(\bv,w_1,w_2) \notin \{\chi_{i}(\bv',w_1',w_2') \mid w_1',w_2' \in V\}\]
    it holds that
    \[\chi_{i+1}(\bv,w_1,w_1) \notin \{\chi_{i+1}(\bv',w_1',w_1') \mid w_1' \in V\}.\]
\end{definition}

\begin{lemma}\label{lem:k:wl:mul:ext}
    Let $V$ be a finite set and let $\bchi = (\chi_0,\dots,\chi_r)$ be a $k$-WL sequence over $V$.
    Then $\bchi$ is extendable.
\end{lemma}

\begin{proof}
    Assume that $\bchi = (\chi_0, \dots, \chi_r)$ is not extendable over $V$.
    Then there exist some index $i \in \{0, \dots, r-1\}$, values $\bv, \bv' \in V^{k-2}$, and $w_1, w_2 \in V$ such that 
    \[\chi_i(\bv, w_1, w_2) \notin \{ \chi_i(\bv', w_1', w_2') \mid w_1', w_2' \in V\} \text{ and } \chi_{i+1}(\bv, w_1, w_1) \in \{ \chi_{i+1}(\bv', w_1', w_1') \mid w_1' \in V\}. \]
    Pick $w_1' \in V$ such that $\chi_{i+1}(\bv, w_1, w_1) = \chi_{i+1}(\bv', w_1', w_1')$.
    In particular, the multiset part of $\chi_{i+1}(\bv, w_1, w_1)$ equals the multiset part of $\chi_{i+1}(\bv', w_1', w_1')$.
    However, by definition the multiset part of $\chi_{i+1}(\bv, w_1, w_1)$ contains the tuple $(c_1, \dots, c_{k-1}, \chi_{i}(\bv, w_1, w_2))$, for some $c_1, \dots, c_{k-1} \in C$.
    Also, the multiset part of $\chi_{i+1}(\bv', w_1', w_1')$ consists of elements from
    \[\{(c'_1, \dots, c_{k-1}', \chi_{i+1}(\bv', w_1', w_2')) \mid c'_1, \dots, c_{k-1}' \in C,  w_2' \in V\}.\]
    Because $\chi_{i+1}(\bv, w_1, w_1) = \chi_{i+1}(\bv', w_1', w_1')$, there is some $w_2' \in V$ such that $\chi_{i}(\bv, w_1, w_2) = \chi_{i}(\bv', w_1', w_2')$ which is a contradiction. 
\end{proof}

\subsection{Reducing to Dimension Two}

Now, we can formulate the main technical result of this section, which allows us to prove \Cref{mtheo:k:dim}.
Let $f(n)$ denote the maximal length of a multiplicative $2$-refinement sequence over an $n$-element set $V$.
Note that \cref{thm:upper-bound-multiplicative} implies $f(n) \leq 5(n-1)$.

\theoHighWL*

\begin{proof}
    For a tuple $\bv = (v_1,\dots,v_{k-2}) \in V^{k-2}$ we define $\ext(\bv) = (v_1,\dots,v_{k-2},v_{k-2},v_{k-2}) \in V^k$ to be the tuple obtained from $\bv$ by repeating the last entry two more times.
    For $i \in \{0,\dots,r\}$ let
    \[\bar \chi_i\colon V^{k-2} \to C\colon \bv \mapsto \chi_i(\ext(\bv)).\]
    It is easy to see that $\bar \chi_{i-1} \succeq \bar \chi_{i}$ for every $i \in [r]$.
    Also, note that $\chi_{i-1}^\bv \succeq \chi_{i}^\bv$, $\chi_i^\bv$ is compatible with equality, and $\chi_i^\bv$ is shufflable for every $\bv \in V^{k-2}$ and $i \in [r]$.

    \begin{claim}
        \label{claim:ext-color}
        Let $i \in \{0,\dots,r-2\}$ and $\bv,\bv' \in V^{k-2}$. Suppose there are $w_1,w_2 \in V$ such that
        \[\chi_i(\bv,w_1,w_2) \notin \{\chi_{i}(\bv',w_1',w_2') \mid w_1',w_2' \in V\}.\]
        Then $\chi_{i+2}(\ext(\bv)) \neq \chi_{i+2}(\ext(\bv'))$.
    \end{claim}
    \begin{claimproof}
        First, we have that
        \begin{equation}
            \label{eq:ext-color-1}
            \chi_{i+1}(\bv,w_1,w_1) \notin \{\chi_{i+1}(\bv',w_1',w_1') \mid w_1' \in V\}
        \end{equation}
        since $\bchi$ is extendable.
        Now, suppose that $\bv = (v_1,\dots,v_{k-2})$ and $\bv' = (v_1',\dots,v_{k-2}')$.
        We claim that
        \begin{equation}
            \label{eq:ext-color-2}
            \chi_{i+1}(\bv,v_{k-2},w_1) \notin \{\chi_{i+1}(\bv',v_{k-2}',w_1') \mid w_1' \in V\}.
        \end{equation}
        Indeed, suppose towards a contradiction that there is some $w_1' \in V$ such that
        \[\chi_{i+1}(\bv,v_{k-2},w_1) = \chi_{i+1}(\bv',v_{k-2}',w_1').\]
        Consider the mapping $\sigma\colon [k] \to [k]$ defined via $\sigma(k-1) = k$ and $\sigma(j) = j$ for all $j \neq k-1$.
        Since $\chi_{i+1}$ is shufflable, we conclude that
        \[\chi_{i+1}(\bv,w_1,w_1) = \chi_{i+1}(\bv',w_1',w_1')\]
        which contradicts \eqref{eq:ext-color-1}.
        Since $\chi_{i+1}$ is compatible with equality, we obtain
        \begin{equation}
            \label{eq:ext-color-3}
            \chi_{i+1}(\bv,v_{k-2},w_1) \notin \{\chi_{i+1}(\bv',w_0',w_1') \mid w_0',w_1' \in V\}
        \end{equation}
        by \eqref{eq:ext-color-2}.
        Finally, since $\bchi$ is extendable, we get that
        \begin{equation*}
            \chi_{i+2}(\bv,v_{k-2},v_{k-2}) \notin \{\chi_{i+2}(\bv',w_0',w_0') \mid w_0' \in V\}
        \end{equation*}
        which implies the claim.
    \end{claimproof}

    \begin{claim}\label{claim:bar:or:v}
        For every $i \in \{1,\dots,r-2\}$ it holds that $\bar \chi_{i-1} \succ \bar \chi_{i+2}$ or there is some $\bv \in V^{k-2}$ such that $\chi_{i-1}^{\bv} \succ \chi_{i}^{\bv}$.
    \end{claim}
    \begin{claimproof}
        Suppose $\chi_{i-1}^{\bv} \equiv \chi_{i}^{\bv}$ for every $\bv \in V^{k-2}$.
        Since $\chi_{i-1} \succ \chi_i$, there is a color class $X \subseteq V^{k}$ that is split, i.e., $c \coloneqq \chi_{i-1}(\bw) = \chi_{i-1}(\bw')$ for every $\bw,\bw' \in X$, but there are $\bw,\bw' \in X$ such that $\chi_{i}(\bw) \neq \chi_{i}(\bw')$.
        Let $X = X_1 \uplus \dots \uplus X_\ell$ denote the partition of $X$ into color classes of $\chi_i$.
        For $j \in [\ell]$ let
        \[\rho(X_j) \coloneqq \{\bv \in V^{k-2} \mid \exists w_1,w_2 \in V\colon (\bv,w_1,w_2) \in X_j\}\]
        denote the projection of $X_j$ to the first $k-2$ components.
        Since $\chi_{i-1}^{\bv} \equiv \chi_{i}^{\bv}$ for every $\bv \in V^{k-2}$, we conclude that $\rho(X_j) \cap \rho(X_{j'}) = \emptyset$ for every $j \neq j' \in [\ell]$.
        Now, pick $\bv \in \rho(X_j)$, $\bv' \in \rho(X_{j'})$ for distinct $j \neq j' \in [\ell]$.
        Then $\bar \chi_{i+2}(\bv) \neq \bar \chi_{i+2}(\bv')$ by Claim \ref{claim:ext-color}.
        Let $w_1,w_2,w_1',w_2' \in V$ such that $\bw = (\bv,w_1,w_2) \in X_j$ and $\bw' = (\bv',w_1',w_2') \in X_{j'}$.
        Then $\chi_{i-1}(\bw) = \chi_{i-1}(\bw')$, and hence, $\bar \chi_{i-1}(\bv) = \bar \chi_{i-1}(\bv')$ since $\chi_{i-1}$ is shufflable. 
        We conclude that $\bar \chi_{i-1} \succ \bar \chi_{i+2}$.
    \end{claimproof}

    Now, let us fix an arbitrary order $<$ on $V$.
    For $j \in [k-2]$ let
    \[S_j \coloneqq \{(v_1,\dots,v_j,v_j,\dots,v_j) \in V_{k-2} \mid v_1 < v_2 < \dots < v_j\}.\]
    Observe that $|S_j| \leq \binom{n}{j}$.
    Let $S \coloneqq \bigcup_{j \in [k-2]} S_j$ and observe that $|S| = \sum_{j = 1}^{k-2} \binom{n}{j}$.

    \begin{claim}\label{claim:v:S}
        Suppose $i \in \{1,\dots,r\}$ such that $\chi_{i-1}^{\bv} \succ \chi_{i}^{\bv}$ for some $\bv \in V^{k-2}$.
        Then there is some $\bv' \in S$ such that $\chi_{i-1}^{\bv'} \succ \chi_{i}^{\bv'}$.
    \end{claim}
    \begin{claimproof}
        Pick $\bv = (v_1, \dots, v_{k-2}) \in V^{k-2}$ such that $\chi_{i-1}^{\bv} \succ \chi_{i}^{\bv}$.
        We show how to obtain a $\bv' \in S$ with $\chi_{i-1}^{\bv'} \succ \chi_{i}^{\bv'}$.
        To this end, let $(v'_1, \dots, v'_j)$ be obtained from $(v_1, \dots, v_{k-2})$ by sorting $(v_1, \dots, v_{k-2})$ and deleting duplicates.
        Let $\bv' = (v'_1, \dots, v'_j, v'_j, \dots v'_j) \in S_j$.
        Now, there is a function $\sigma \colon [k-2] \to [k-2]$ with $(v_{\sigma(1)}, \dots, v_{\sigma(k-2)}) = (v'_1, \dots, v'_j, v'_j, \dots v'_j)$.
        Furthermore, there is a function $\sigma' \colon [k-2] \to [j]$ with $(v'_{\sigma'(1)}, \dots, v'_{\sigma'(k-2)}) = (v_1, \dots, v_{k-2})$.

        Since $\chi_{i-1}^{\bv} \succ \chi_{i}^{\bv}$, there are $w_1, w_2, w'_1, w'_2$ with 
        \begin{align}
            \chi_{i-1}(v_1, \dots, v_{k-2}, w_1, w_2) ~&=~  \chi_{i-1}(v_1, \dots, v_{k-2}, w'_1, w'_2) \quad \text{ and } \label{eq:chi:v:1} \\
            \chi_{i}(v_1, \dots, v_{k-2}, w_1, w_2) ~&\neq~  \chi_{i}(v_1, \dots, v_{k-2}, w'_1, w'_2). \label{eq:chi:v:2}
        \end{align}
        We apply shufflablity to \eqref{eq:chi:v:1} which yields
        \[\chi_{i-1}(v_{\sigma(1)}, \dots, v_{\sigma(k-2)}, w_1, w_2) = \chi_{i-1}(v_{\sigma(1)}, \dots, v_{\sigma(k-2)}, w'_1, w'_2)\]
        and is equivalent to $\chi_{i-1}(\bv', w_1, w_2) = \chi_{i-1}(\bv', w'_1, w'_2)$.
        
        Also, we obtain that $\chi_{i}(\bv', w_1, w_2) \neq \chi_{i}(\bv', w'_1, w'_2)$.
        Indeed, if not, we can use $\sigma'$ to obtain
        \[\chi_{i}(v'_{\sigma'(1)}, \dots, v'_{\sigma'(k-2)}, w_1, w_2) = \chi_{i}(v'_{\sigma'(1)}, \dots, v'_{\sigma'(k-2)}, w'_1, w'_2),\]
        which is equivalent to $\chi_{i}(\bv, w_1, w_2) = \chi_{i}(\bv, w'_1, w'_2)$ and contradicts \eqref{eq:chi:v:2}.
        Hence, $\chi_{i}(\bv', w_1, w_2) \neq \chi_{i}(\bv', w'_1, w'_2)$ which together with $\chi_{i-1}(\bv', w_1, w_2) = \chi_{i-1}(\bv', w'_1, w'_2)$ yields $\chi_{i-1}^{\bv'} \succ \chi_{i}^{\bv'}$.
    \end{claimproof}

    We now show that $r = |\bchi| \leq 3 \cdot n^{k-2} + f(n) \cdot |S|$ which yields the theorem.
    To this end, let
    \[\bar{I} = \{i \in [r] \mid \bar{\chi}_{i-1} \succ \bar{\chi}_{i}\}\]
    and, for $\bv \in S$, let
    \[I^{\bv}= \{i \in [r] \mid \chi^\bv_{i-1} \succ \chi^\bv_{i}\}.\]
    Combing \cref{claim:bar:or:v} with \cref{claim:v:S}, for every $i \in [r-2]$, we have that $\bar \chi_{i-1} \succ \bar \chi_{i+2}$ or $\chi_{i-1}^{\bv} \succ \chi_{i}^{\bv}$ for some $\bv \in S$.
    In the former case, we conclude that $\{i,i+1,i+2\} \cap \bar{I} \neq \emptyset$, and in the latter case we have $i \in I^{\bv}$ for some $\bv \in S$.
    It follows that
    \[r-2 \leq 3 \cdot |\bar I| + \sum_{\bv \in S} |I^{\bv}|.\]
    Since each $\bar{\chi}_i$ is a coloring of $V^{k-2}$, we obtain that
    \[|\bar{I}| \leq n^{k-2} - 1.\]
    Also, for $\bv \in S$, after removing duplicates from the sequence $(\chi_0^{\bv},\dots,\chi_r^{\bv})$, we obtain a multiplicative $2$-refinement sequence.
    It follows that $|I^{\bv}| \leq f(n)$.
    Overall, we obtain that
    \[r \leq 3 \cdot |\bar I| + 2 + \sum_{\bv \in S} |I^{\bv}| \leq 3 \cdot (n^{k-2} - 1) + 2 + |S| \cdot f(n) \leq 3 \cdot n^{k-2} + f(n) \cdot \sum_{j = 1}^{k-2} \binom{n}{j}\]
    as desired.
\end{proof}

\mtTwo*

\begin{proof}
    According to \cref{lemma:k:WL:refinement,lem:k:wl:mul:ref,lem:k:wl:mul:ext}, the sequence $\bchi = (\chi_0, \dots, \chi_r)$ is multiplicative and extendable $k$-refinement sequence over $V$. 
    Thus, \cref{theo:upper:bound:k} together with \cref{thm:upper-bound-multiplicative} yields
    \[|\bchi| \leq 3 \cdot n^{k-2} + 5(n-1) \cdot \sum_{j = 1}^{k-2} \binom{n}{j}.\]
    Now, the result is a direct consequence of $\sum_{j = 1}^{k-2}  \binom{n}{j} \leq 5 \cdot \frac{n^{k-2}}{(k-2)!}$ (see \cref{lem:sum}).
\end{proof}

\section*{AI Statement}

An initial version of the proof of \Cref{lemma:lambda:sum,lem:mul:bound:beta} was found by ChatGPT 5.6 Sol.
Also, the same tool was used to analyze several recursive cost functions, including the function $C$ defined in \eqref{eq:def-c}.
We stress that all proofs in this version are written and checked by the authors.

\bibliographystyle{plainurl}
\bibliography{refs}

\clearpage

\appendix

\section{Sums of Binomial Coefficients}

\begin{lemma}\label{lem:sum}
    For every $n \geq k \geq 1$ it holds that 
    \[\sum_{j = 1}^k \binom{n}{j} \leq 5 \cdot \frac{n^{k}}{k!}.\]
\end{lemma}
\begin{proof}
    We first show that for $2k \leq n$, we obtain $\sum_{j = 1}^k \binom{n}{j} \leq 2 \cdot  \frac{n^{k}}{k!}$.
    To this end, observe that
    \begin{align*}
        \binom{n}{j} \leq \frac{n^j}{j!} = \frac{n^k}{j!} \cdot \frac{1}{n^{k-j}} \leq \frac{n^k}{k!} \cdot \frac{k^{k-j}}{n^{k-j}} = \frac{n^k}{k!} \cdot \left(\frac{k}{n}\right)^{k-j} \leq \frac{n^k}{k!} \cdot \frac{1}{2^{k-j} },
    \end{align*}
    where the last step uses that $2k \leq n$.
    Using this together with the properties of geometric series, we obtain
    \[\sum_{j = 1}^k \binom{n}{j}  \leq \frac{n^k}{k!} \cdot \sum_{j = 1}^k \frac{1}{2^{k-j} } \leq 2 \cdot \frac{n^k}{k!}.  \]
    Finally, for $2k > n$, we obtain
    \[\sum_{j = 1}^k \binom{n}{j} \leq \sum_{j = 1}^{\lfloor n/2 \rfloor} \binom{n}{j} + \binom{n}{\lceil n/2 \rceil} + \sum_{j = 1 + \lceil n/2 \rceil}^{n} \binom{n}{j} = \binom{n}{\lceil n/2 \rceil}  + 2 \cdot \sum_{j = 1}^{\lfloor n/2 \rfloor} \binom{n}{j} \leq  5 \cdot \frac{n^k}{k!}, \]
    using that $n^k / k!$ is monotone increasing with respect to $k \leq n$.
\end{proof}

\end{document}